\documentclass[jair,twoside,11pt,theapa]{article}
\usepackage{jair, theapa, rawfonts}

\ShortHeadings{Practical and Parallelizable Algorithms for SMCC}
{Chen \& Kuhnle}

\usepackage{xspace}
\usepackage{amsthm}
\usepackage{amsmath}
\usepackage{thmtools}
\usepackage{booktabs}
\usepackage{algorithm}
\usepackage[noend]{algpseudocode}
\usepackage{xcolor}
\usepackage{subfigure}
\usepackage{thm-restate}
\usepackage{parskip}
\usepackage{makecell}
\usepackage{threeparttable}
\usepackage{epsfig}
\usepackage{amssymb}
\usepackage{graphicx}
\usepackage{ulem}
 
\newtheorem{theorem}{Theorem}
\newtheorem{lemma}[theorem]{Lemma}

\newtheorem{definition}[theorem]{Definition}

\newtheorem*{definition*}{Definition}

\newcommand{\ie}{\textit{i.e.}\xspace}
\newcommand{\eg}{\textit{e.g.}\xspace}

\newcommand{\ex}[1]{ \mathbb{E} \left[ #1 \right] }
\newcommand{\exc}[2]{ \mathbb{E} \left[ #1 \, | \; #2 \right] }
\newcommand{\prob}[1]{ Pr \left( #1 \right) }

\newcommand{\epsi}[0]{ \varepsilon }
\newcommand{\reals}[0]{ \mathbb{R}^+ }

\newcommand{\ff}[1]{ f \left( #1 \right) }

\newcommand{\marge}[2]{\Delta \left( #1 \, | \, #2 \right) }

\newcommand{\argmax}{\arg\, \max}

\newcommand{\oh}[1]{\mathcal{O}\left( #1 \right)}
\newcommand\numberthis{\addtocounter{equation}{1}\tag{\theequation}}

\renewcommand{\restriction}{\mathord{\upharpoonright}}

\newcommand{\unc}{\textsc{UnconstrainedMax}\xspace}

\newcommand{\threseq}{\textsc{ThreshSeq}\xspace}
\newcommand{\threseqama}{\textsc{TS-AMA-v1}\xspace}
\newcommand{\tsbin}{\textsc{TS-AMA-v2}\xspace}
\newcommand{\thresam}{\textsc{Threshold-Sampling}\xspace}
\newcommand{\adapst}{\textsc{AdaptiveSimpleThreshold}\xspace}
\newcommand{\adaptg}{\textsc{AdaptiveThresholdGreedy}\xspace}
\newcommand{\sm}{\textbf{SMCC}\xspace}
\newcommand{\tp}{\textsc{Threshold}\xspace}

\newcommand{\uni}{\mathcal N}

\newcommand{\atgRatio}{\frac{e-1}{e(2+\alpha)-\alpha}-\epsi}
\newcommand{\opt}{\textsc{OPT}\xspace}

\newcommand{\reducedmean}{\textsc{ReducedMean}\xspace}

\newcommand{\algOnefullname}{\textsc{AdaptiveSimpleThreshold}\xspace}
\newcommand{\algTwofullname}{\textsc{AdaptiveThresholdGreedy}\xspace}
\newcommand{\anm}{\textsc{AdaptiveNonmonotoneMax}\xspace}
\newcommand{\atg}{\textsc{AST}\xspace}
\newcommand{\latg}{\textsc{ATG}\xspace}
\newcommand{\iter}{\textsc{IteratedGreedy}\xspace}
\newcommand{\frg}{\textsc{FastRandomGreedy}\xspace}

\newcommand{\threshgreedy}{\textsc{ThresholdGreedy}\xspace}
\newcommand{\park}{\textsc{ParCardinal}\xspace}
\newcommand{\parkone}{\textsc{ParCardinal}(v1)\xspace}
\newcommand{\parktwo}{\textsc{ParCardinal}(v2)\xspace}
\newcommand{\iterratio}{\frac{e - 1}{e(2 + \alpha) - \alpha}}
\newcommand{\red}[1]{#1} 
\newcommand{\revtwo}[1]{#1} 

\begin{document}
\title{Practical and Parallelizable Algorithms for Non-Monotone Submodular Maximization with Size Constraint}


  \author{\name Yixin Chen \email chen777@tamu.edu \\
       \addr Department of Computer Science \& Engineering\\
        Texas A\&M University\\
         College Station, TX
       \AND
       \name Alan Kuhnle \email kuhnle@tamu.edu \\
       \addr Department of Computer Science \& Engineering\\
         Texas A\&M University\\
         College Station, TX}


\maketitle
\begin{abstract}
  We present combinatorial and parallelizable algorithms
  for \revtwo{the} maximization of a submodular function, not necessarily
  monotone, with respect to a size constraint.
  We improve the best approximation factor achieved
  by an algorithm that has optimal adaptivity and
  nearly optimal query
  complexity to \red{$1/6 - \epsi$,
  and even further to $0.193 - \epsi$ by increasing the adaptivity by a factor of $\oh{\log (k)}$}. 
  The conference version
  of this work mistakenly employed a subroutine that does
  not work for non-monotone, submodular functions.
  In this version, we propose a fixed and improved subroutine
  to add a set with high average marginal gain,
  \threseq, 
  which returns a solution in $\oh{\log(n)}$
  adaptive rounds with high probability.
    Moreover, we provide two approximation algorithms.
    The first has approximation
    ratio $1/6 - \epsi$, adaptivity
    $\oh{\log (n)}$, and query complexity $\oh{ n \log (k)}$,
    while the second has approximation ratio $0.193 - \epsi$,
    adaptivity $ \oh{\log (n) \log(k) }$, and query complexity $\oh{n \log (k)}$.
    Our algorithms are empirically 
    validated to use a low number of adaptive rounds and total queries
    while obtaining solutions with high objective value in comparison with
    state-of-the-art approximation algorithms, including continuous algorithms
    that use the multilinear extension. 
    \end{abstract}



\section{Introduction} \label{sec:intro}
A nonnegative set function $f:2^{\mathcal N} \to \reals$, defined on all subsets
of a ground set $\mathcal N$ of size $n$,
is \textit{submodular}
if for all $A, B \subseteq \mathcal N$,
$f(A) + f(B) \ge f(A \cup B) + f( A \cap B )$.
\red{A set function is monotone if $A \subseteq B$ implies
$f(A) \le f(B)$.}
Submodular set functions naturally arise in
many learning applications, 
including data summarization \cite{Simon2007,Sipos2012,Tschiatschek2014,Libbrecht2017}, viral
marketing \cite{Kempe2003,Hartline2008}, and recommendation systems \cite{El-Arini2011}. 
Some applications yield submodular functions
that are not monotone: for example, image summarization with
diversity \cite{Mirzasoleiman2016} or revenue maximization on
a social network \cite{Hartline2008}.
In this work, we study the maximization of
a (not necessarily monotone) submodular function subject to a cardinality constraint;
that is, given submodular function $f$ and integer $k$, determine
$\argmax_{|S| \le k} f(S)$ (\sm). Access to $f$
is provided through a value query oracle, which
when queried with the set $S$ returns the value $f(S)$.

As the amount of data in applications has exhibited
exponential growth in recent years 
(\eg the growth of social networks \cite{Mislove2008}
or genomic data \cite{Libbrecht2017}), it is 
necessary to design algorithms for \sm that can scale to
these large datasets. 
One aspect of algorithmic efficiency is the \textit{query complexity},
the total number
of queries to the oracle for $f$. 
Since evaluation of $f$
is often expensive, the queries to $f$ often dominate the
runtime of an algorithm. In addition to low query complexity,
it is necessary to
design algorithms that parallelize well  to take advantage of
modern computer architectures. To quantify the degree
of parallelizability of an algorithm,
the \textit{adaptivity} or \textit{adaptivive complexity} of an algorithm
is the minimum number of sequential rounds such that
in each round the algorithm makes $\oh{\text{poly}(n)}$
independent queries to the evaluation oracle.
The lower the adaptivive complexity of an algorithm,
the more suited the algorithm is to parallelization,
as within each adaptive round, the queries to $f$
are independent and may be easily parallelized. 

The design of algorithms with nontrivial
adaptivity for \sm when $f$ is monotone 
was initiated by {\citeA{Balkanski2018}},
who also prove a lower bound of $\Omega ( \log( n) / \log \log (n) ) $
adaptive rounds to achieve a constant approximation ratio. Recently, much
work has focused on the design of adaptive algorithms for \sm
with (not necessarily monotone) submodular functions,
as summarized in Table~\ref{table:cmp}.
However, although many algorithms with low adaptivity have been proposed,
most of these algorithms exhibit at least a quadratic dependence
of the query complexity on the size $n$ of the ground set, for $k = \Omega(n)$.
For many applications, instances have grown too large
for quadratic query complexity to be practical.
Therefore, it is necessary
to design adaptive algorithms that also have
nearly linear query complexity.
An algorithm in prior literature that
meets this requirement is the
algorithm developed by {\citeA{Fahrbach2018a}}, which has $\oh{n \log (k)}$
query complexity and $\oh{\log (n)}$ adaptivity.
However, the 
approximation ratio stated in {\citeA{Fahrbach2018a}}
for this algorithm does not hold,
as discussed in Section \ref{sec:related_work} 
and Appendix \ref{sec:counterexample}.
\revtwo{During our revision of this paper, 
\citeA{fahrbach2018non} fixed it,
ensuring that the approximation ratio holds now.}
\begin{table*}[t] \centering 
  \begin{threeparttable}
  \begin{minipage}{\textwidth}
  \caption{Adaptive algorithms for \sm where objective $f$ is not 
  necessarily monotone. 
  \revtwo{We consider three metrics here.
  ``Approximation Ratio'' reflects the accuracy of the algorithm, where a higher value signifies greater accuracy.
  ``Adaptivity'' measures the algorithm's parallelizability,
  with a lower value indicating higher parallelizability.
  ``Queries'' represents the total number of queries to the oracle for $f$, dominating the algorithm's runtime. 
  A lower query count implies faster algorithmic performance. 
  Both the adaptivity and query complexity values presented in this table are asymptotic.
  }} \label{table:cmp}
\begin{tabular}{l|l|l|l} 
  \toprule
Reference         & Approximation Ratio & Adaptivity & Queries \\
  \midrule
  {\shortciteA{Buchbinder2015a}} & $1/e - \epsi \revtwo{\approx 0.367-\epsi}$ & $k$ & ${n}$ \\
  \midrule
  {\shortciteA{Balkanskia}} & $1/(2e) - \epsi\revtwo{\approx 0.183-\epsi}$ & $ \log^2(n) $ & $ OPT^2 n \log^2(n) \log(k) $\\
  \midrule
  {\shortciteA{Chekuri2019a}} & $3 - 2 \sqrt{2} - \epsi\revtwo{\approx 0.171-\epsi}$ & $\log^2(n)$ & $ nk^4 \log^2(n)$ \\
  \midrule
  {\shortciteA{Ene2020}} & $1/e - \epsi\revtwo{\approx 0.367-\epsi}$ & $ \log(n)$ & $ nk^2 \log^2(n) $ \\
  \midrule
  {\shortciteA{fahrbach2018non}}& $0.039 - \epsi$
  & $ \log (n)$ & ${n \log (k)}$ \\
  \midrule
  {\shortciteA{amanatidis2021submodular}}& $0.172-\epsi$ 
  & \makecell[l]{$\log(n)$ \\ $\log(n)\log(k)$} 
  & \makecell[l]{$nk\log(n)\log(k)$ \\ $n\log(n)\log^2(k)$} \\
  \midrule
  Theorem~\ref{thm:atg} (\atg)            & $1/6 - \epsi\revtwo{\approx 0.166-\epsi}$ & ${\log (n)}$ & $n \log (k)$ \\
  Theorem~\ref{thm:latg} (\latg)      & $0.193 - \epsi$ & $ \log(n)\log(k)$ & $n \log (k)$ \\ 
  \bottomrule
\end{tabular}
\end{minipage}
\end{threeparttable}
\end{table*}

\textbf{Contributions.}
In this work, we propose two fast,
combinatorial algorithms for \sm:
the $(1/6 - \epsi)$-approximation algorithm 
\algOnefullname (\atg)
with adaptivity $\oh{\log (n)}$ and query complexity $\oh{n \log (k)}$;
and the $(0.193 - \epsi)$-approximation algorithm \algTwofullname (\latg) with
adaptivity $\oh{\log (n) \log (k)}$ and query complexity $\oh{n \log (k)}$. 

The above algorithms both employ a lowly-adaptive
subroutine to add
multiple elements that satisfy a given marginal gain,
on average.
The conference version \cite{kuhnle2021nearly} of this paper
used
the \thresam subroutine of
{\citeA{Fahrbach2018,Fahrbach2018a}} for this purpose.
However, the theoretical guarantee (Lemma 2.3 of {\citeA{Fahrbach2018a}}) for non-monotone functions does not hold 
\revtwo{due to a bug that has since been fixed in~\citeA{fahrbach2018non}.}
In Appendix~\ref{sec:counterexample}, we give a counterexample to the performance guarantee of \thresam.
In this version,
we introduce a new threshold subroutine \threseq,
which not only fixes the problem that \thresam faced,
but achieves its guarantees with high probability
as opposed to in expectation; the high probability guarantees
simplify the analysis of our approximation algorithms that
rely upon the \threseq subroutine. 



Our algorithm \atg uses a double-threshold procedure to obtain
its ratio of $1/6 - \epsi$. Our second algorithm \latg 
is a low-adaptivity modification of the algorithm of {\citeA{Gupta2010a}}, for which 
we improve the ratio from $1/6$ to {0.193} through a novel analysis.
Both of our algorithms use the low-adaptivity, threshold sampling procedure
\threseq and a subroutine for 
unconstrained maximization of a submodular function \cite{Feige2011,Chen2018b}
as components.
More details are given in the related work
discussion below and in Section~\ref{sec:latg}.

The new \threseq does not rely on sampling to achieve
concentration bounds, which
significantly
improves the practical efficiency of our algorithms
over the conference version \red{\cite{kuhnle2021nearly}}.
Empirically, we demonstrate that both of our algorithms achieve
superior objective value to current state-of-the-art algorithms while using a small
number of queries and adaptive rounds on two applications of \sm. 
\subsection{Related Work}
\label{sec:related_work}
\textbf{Theshold Procedures.}
A recurring subproblem of \sm (and other submodular optimization problems)
is to 
\red{add to a candidate solution $S$ those elements $x$ of the ground set $\uni$}
that give a marginal gain of
at least $\tau$, for some constant threshold $\tau$. 
To solve this subproblem, 
the algorithm \thresam is proposed in 
{\citeA{Fahrbach2018}} for monotone submodular functions
and applied in {\citeA{Fahrbach2018a}}
and the conference version of this work \cite{kuhnle2021nearly} as subroutines
for non-monotone \sm. However, theoretical guarantee
(Lemma 2.3 of {\citeA{Fahrbach2018a}}) does not hold when the objective
function is non-monotone. Counterexamples and pseudocode for \thresam are given in Appendix~\ref{sec:counterexample}.
\revtwo{A recent work by~\citeA{fahrbach2018non} has modified the \thresam algorithm and fixed the problem discussed above.}


Two alternative solutions to the non-monotone threshold problem were proposed 
in {\citeA{amanatidis2021submodular}} for the case of non-monotone,
submodular maximization subject to a knapsack constraint.
Due to the complexity of the constraints, the thresholding procedures
in {\citeA{amanatidis2021submodular}} have a high time complexity and
require
$\oh{n^2}$ query calls within one iteration
even when restricted to {a} size constraint.
Although a variant with binary search is proposed to get 
fewer queries, the sequential binary search worsens
the adaptivity of the algorithm.

In this work, we propose the
\threseq algorithm (Section \ref{sec:ts})
that fixes the problems of \thresam and runs
in \red{$\oh{n}$ queries and $\oh{ \log n}$ adaptive rounds.}
We solve these problems
by \red{introducing two sets found by the algorithm:}
an auxilliary set $A$ separate from the solution
set $A'$ found by \threseq \red{that solves \tp (Def.~\ref{def:tp}) separately}. 
The algorithm maintains that
$A' \subseteq A$, and the larger set is used for filtering
from the ground set, while the smaller set maintains desired
bounds on the average marginal gain. 

\textbf{Algorithms with Low Adaptivive Complexity.}
Since the study of parallelizable algorithms for
submodular optimization was initiated by
{\citeA{Balkanski2018}}, there have been a
number of $\oh{\log n}$-adaptive algorithms designed
for \sm. When $f$ is monotone, adaptive algorithms
that obtain the optimal ratio \cite{Nemhauser1978a} of $1 - 1/e - \epsi$
have been designed by {\citeA{Balkanski,Fahrbach2018,Ene,chen2021best}}.
Of these, the algorithm of {\citeA{chen2021best}} also has
the state-of-the-art
sublinear adapativity and linear query complexity.

However, when the function $f$ is not monotone, the
best approximation ratio with polynomial query complexity
for \sm is unknown, but falls within the range $[0.385, 0.491]$
\cite{Buchbinder2016,Gharan2011a}. For \sm,
algorithms with nearly optimal adaptivity have been
designed by {\citeA{Balkanskia,Chekuri2019a,Ene2019b,Fahrbach2018a,amanatidis2021submodular}};
for the query complexity and approximation factors of
these algorithms, see Table~\ref{table:cmp}.
Of these, the best approximation ratio of $(1/e - \epsi) \approx 0.368$ 
is obtained by the algorithm of
{\citeA{Ene2020}}.
However, this algorithm requires access to an oracle for
the gradient of the continuous extension of a submodular
set function, which requires $\Omega (nk^2 \log^2 (n) )$ 
queries to sufficiently approximate. 
The practical performance
of the algorithm of {\citeA{Ene2020}} is 
investigated in our empirical evaluation of
Section~\ref{sec:exp}.
Other than the algorithms of {\citeA{Fahrbach2018}} and {\citeA{amanatidis2021submodular}}, 
all parallelizable
algorithms exhibit a runtime of at least quadratic dependence on $n$.
In contrast, our algorithms have query complexity of 
$\oh{n \log k}$ and have $\oh{\log n}$ or $\oh{\log^2 n}$
adaptivity.  

After the conference version \cite{kuhnle2021nearly} of this paper,
{\citeA{amanatidis2021submodular}} proposed a parallelizable algorithm,
\park, for knapsack constraints,
which is the first constant factor approximation with optimal
adaptivive complexity.
In the paper, \park is directly applied to cardinality constraints.
It achieves a $0.172-\epsi$ ratio with two different variants:
one has $\oh{\log(n)}$ adaptive rounds and $\oh{nk\log(n)\log(k)}$ queries;
another one has $\oh{\log(n)\log(k)}$ adaptive rounds and $\oh{n\log(n)\log^2(k)}$ queries.
Compared to our nearly linear algorithms, 
the first variant of \park requires total queries with more than quadratic dependence on $n$; 
and the second variant gets a worse approximation ratio and worse number of queries than our algorithm (\latg) with the same adaptivity.

\textbf{The \iter Algorithm.}
Although the standard greedy algorithm performs arbitrarily
badly for \sm,
{\citeA{Gupta2010a}} showed that multiple repetitions of the
greedy algorithm, combined with an approximation for
the unconstrained maximization problem, yields an approximation
for \sm. 
Specifically, {\citeA{Gupta2010a}} provided
the \iter algorithm, which
achieves an approximation ratio of $1/6$
for \sm when
the $1/2$-approximation of {\citeA{Naor2012}} is used
for the unconstrained maximization subproblems.
Our algorithm \algTwofullname uses \threseq combined with
the descending thresholds technique of {\citeA{Badanidiyuru2014}} to
obtain an adaptive version
of \iter, as described
in Section~\ref{sec:latg}. Pseudocode for \iter is given in Appendix~\ref{apx:iter},
where an improved ratio of $\approx$0.193 is proven for this algorithm; we
also prove the ratio of nearly $0.193$ 
for our adaptive algorithm \latg in Section~\ref{sec:latg}.


\subsection{Preliminaries} \label{sec:prelim} 
A submodular set function defined on all subsets of ground set 
$\mathcal N$ is denoted by $f$. 
The marginal gain of adding an element $x$ to a set 
$S$ is denoted by $\marge{x}{S} = f( S \cup \{ x \} ) - f(S)$. 
Let $\opt = \max_{|S|\le k}f(S)$\red{, the optimal value of the SMCC 
problem for ground set $\uni$ and size constraint k.}
The restriction
of $f$ to all subsets of a set $S \subseteq \mathcal N$ is denoted by
$f \restriction_{S}$. 
Next, we describe two subproblems both
of our algorithms need to solve: namely, 
unconstrained maximization subproblems and
a threshold sampling subproblem.
For both of these subproblems, procedures with low adaptivity are needed.

\textbf{The Unconstrained Maximization Problem.} 
The first subproblem is unconstrained maximization
of a submodular function.
When the function
$f$ is non-monotone, the problem of maximizing $f$ without
any constraints is NP-hard \cite{Feige2011}.
Recently, {\citeA{Chen2018b}} developed an algorithm that achieves
nearly the optimal ratio of $1/2$ with constant adaptivity, as summarized in the
following theorem. 
\begin{theorem}[{\citeA{Chen2018b}}] \label{lemm:unc}
  For each $\epsi > 0$,
  there is an algorithm that
  achieves a $(1/2 - \epsi)$-approximation
  for unconstrained submodular maximization using
  $\oh{\log (1/\epsi ) / \epsi}$ adaptive rounds 
  and $\oh{n \log^3 (1/ \epsi ) / \epsi^4}$ evaluation
  oracle queries.
\end{theorem}
\noindent To achieve the approximation factor listed for
our algorithms in
Table~\ref{table:cmp}, the algorithm of {\citeA{Chen2018b}} is employed
for unconstrained maximization subproblems.

\textbf{The \tp Problem.}
The second subproblem is the following:
\begin{definition}[\tp]\label{def:tp}
  Given a threshold $\tau \in \mathbb R$ and integer $k$,
choose a set $S$ such that 1) $f(S) \ge \tau |S|$; 2)
if $|S| < k$, then
for any $x \not \in S$, $\marge{x}{S} < \tau$.
\end{definition}
Algorithms that can use a solution to this
subproblem occur frequently, and so
multiple algorithms in the literature
for this subproblem
have been formulated
\cite{Fahrbach2018,Balkanski2018c,Kazemi2019,amanatidis2021submodular,chen2021best}.
We want a procedure that can solve
$\tp$ with the following three properties:
1) \red{\revtwo{using} $\oh{n}$ queries of the submodular function}; 2) in $\oh{\log n}$
adaptive rounds; 3) the function $f$ is non-monotone.

None of the prior algorithms satisfy our
requirements, since
the procedures in {\citeA{Fahrbach2018,Kazemi2019,chen2021best}} only work when
the submodular function is monotone;
and the two procedures in {\citeA{amanatidis2021submodular}}
have either $\oh{n^2 \log(n)}$ queries or $\oh{\log^2(n)}$ adaptivity.
Moreover, in both {\citeA{Fahrbach2018}}
and {\citeA{amanatidis2021submodular}}, the
procedures for $\tp$ only guarantee
$\ex{ f(S) } \ge \tau |S|$.

\red{In this paper, we propose \threseq, 
an algorithm that makes $\oh{n}$ query calls and has $\oh{\log n}$
adaptive rounds to solve \tp.
However, this algorithm does not exactly solve $\tp$. 
Instead,
it returns two sets that
solve each of the questions in \tp,
which is enough for our algorithms.}

\textbf{Organization.}
In Section~\ref{sec:ts}, we introduce our threshold sampling
algorithm: \threseq.
Then, in Sections~\ref{sec:atg} and~\ref{sec:latg}, 
we analyze our algorithms
using the \threseq and \unc procedures.
Our empirical evaluation is reported in Section~\ref{sec:exp}
with more discussions in Appendix~\ref{apx:exp}.


\section{The \threseq Algorithm} \label{sec:ts}
In this section, we introduce the linear and highly parallelizable
threshold sampling algorithm
\threseq (Alg.~\ref{alg:threshold}).
\red{\threseq takes as input oracle $f$, constraint $k$, error rate $\epsi$, threshold $\tau$,
and \revtwo{failure probability parameter} $\delta$ which reflects the success probability.}
This algorithm has logarithmic adaptive rounds and linear query calls
with high probability.
Rather than directly solving \tp \red{(Def.~\ref{def:tp})} \red{with one solution set,
it returns two relevant sets that deal with the two properties separately.} 
\begin{algorithm}
\caption{\revtwo{A general framework of threshold sampling algorithms}}
\label{alg:ts_framework}
\begin{algorithmic}[1]
	\Procedure{}{$f, \uni, k$}
		\State \textbf{Input:} evaluation oracle $f:2^{\mathcal N} \to \reals$, constraint $k$, error $\epsi$, threshold $\tau$
		\State Initialize $V\gets \uni$, $A\gets \emptyset$ \label{line:ts_framework_init}
		\While{$|V|> 0$}
			\State $V\gets \left\{x \in V: \marge{x}{A} \ge \tau \text{ and }A\cup \{x\} \text{ feasible}\right\}$ \label{line:ts_framework_filter}
			\State $T \gets$ a subset of $V$ \Comment{make a decision on selecting a good subset from $V$}\label{line:ts_framework_decide}
			\State $A\gets A \cup T$\label{line:ts_framework_add}
		\EndWhile
		\State \textbf{return} $A$
	\EndProcedure
\end{algorithmic}
\end{algorithm}

\subsection{Algorithm Overview}
\red{The state-of-the-art threshold sampling algorithms, 
whether for monotone or non-monotone functions, 
share a common structure \revtwo{(Alg.~\ref{alg:ts_framework})} that works as follows:
1) The algorithm initializes a candidate set $V$ with the whole ground set $\uni$ and
an empty solution set $A$ \revtwo{(Line~\ref{line:ts_framework_init})};
2) During each iteration, it filters out elements in the candidate set $V$ 
that either make negligible contributions to
$A$ or violate the given constraint,
and then selects a prefix of $V$ to add to $A$ \revtwo{(Line~\ref{line:ts_framework_filter}-\ref{line:ts_framework_add})};
3) Then, the algorithm repeats the last step until
the candidate set $V$ is empty.
The difference between those algorithms lies in how they select the prefix in Step (2) \revtwo{on Line~\ref{line:ts_framework_decide}}.
\thresam in~\citeA{Fahrbach2018} applies a random sampling procedure for
each prefix considered at that iteration.
Threshold sampling algorithms in~\citeA{Balkanski2018c} and~\citeA{amanatidis2021submodular}
explicitly check all the candidate elements for a given prefix.
Later, ~\citeA{Kazemi2019} and~\citeA{chen2021best} proposed threshold sampling algorithms
that work by performing a uniformly random permutation of elements
and making the decision after querying once for each prefix.
This makes them comparably much more practical in performance
and demonstrates that multiple query calls of a given prefix are redundant.
Subsequently, we are able to keep a solution with the same threshold and 
fewer query calls.}

\red{To efficiently obtain large sequences of elements with gains above
$\tau$, an approach inspired by monotone threshold sampling algorithms
in\revtwo{~\citeA{Kazemi2019} and~\citeA{chen2021best}} is proposed.
As discussed above,}
these algorithms work by adaptively adding sequences of elements
to a set $A$,
where the sequence has been checked in parallel
to have at most an $\epsi$ fraction of the sequence
failing the marginal gain condition.
A uniformly 
random permutation of elements is considered,
where the average marginal
gain being below $\tau$ is detected by a high proportion
of failures in the sequence.
\red{This step leads to a constant fraction
of elements being filtered out at the next iteration
with high probability.
When combined with an exponentially decreasing candidate set 
and a constant number of adaptive rounds for each iteration, 
these algorithms achieve logarithmic adaptivity and linear query complexity.}

\begin{algorithm}[h]
	\caption{A parallelizable threshold algorithm for threshold $\tau$}
	\label{alg:threshold}
	\begin{algorithmic}[1]
	\Procedure{\threseq}{$f, \mathcal N, k, \delta, \epsi, \tau$}
	\State \textbf{Input:} evaluation oracle $f:2^{\mathcal N} \to \reals$, constraint $k$, \revtwo{failure probability parameter} $\delta$, error $\epsi$, threshold $\tau$
	\State Initialize $A \gets \emptyset$, $A' \gets \emptyset$, $V \gets \mathcal N$, 
	$\ell = \lceil 4\left(\frac{2}{\epsi}\log(n)+\log\left(\frac{n}{\delta}\right)\right) \rceil$ 
	\For{ $j \gets 1$ to $\ell$}  \Comment{Sequential \textbf{for} loop}
		\State Update $V \gets \{ x \in V : \marge{x}{A} \ge \tau \}$ \label{line:threshold-filtering} \Comment{Filtering step w.r.t. $A$}
		\If{ $|V| = 0$ } 
			\State \textbf{return} $A,A'$ \label{line:tsreturn}
		\EndIf
		\State $V \gets$ \textbf{random-permutation}$(V)$ \label{line:threshold-permute}
		\State $s \gets \min \{k-|A|, |V|\}$
		\State $B[1:s] \gets [\textbf{none},\cdots,\textbf{none}]$
		\For{$i \gets 1 $ to $s$ in parallel} \Comment{Parallel gain computation}
			\State $T_{i-1} \gets \{v_1, v_2, \ldots, v_{i-1}\}$ 
			\State \textbf{if} $ \marge{v_i}{A\cup T_{i-1}} \geq  \tau $
				\textbf{then} $B[i] \gets \textbf{true}$ \label{line:threshold-if}
			\State \textbf{elif} $ \marge{v_i}{A\cup T_{i-1}} <  0$
				\textbf{then} $B[i] \gets \textbf{false}$
		\EndFor
		\State $i^* \gets \max\{i:\# \text{\textbf{true}s in } B[1:i] \ge (1-\epsi)i\}$ \Comment{Detection of good filtering next iteration} \label{line:threshold-select_of_istar}
		\State $A \gets A \cup T_{i^*}$ \Comment{$A$ gets all elements}
		\State \revtwo{$A' \gets A' \cup \{V[i]: 1\le i \le i^*, B[i] \neq \textbf{false}\} $} \label{line:threshold-Aprime} \Comment{$A'$ only gets nonnegative-gain elements}
		\If{ $|A| = k$ }
			\State \textbf{return} $A,A'$ 
		\EndIf
	\EndFor
	\State \textbf{return} \textit{failure}
	\EndProcedure
\end{algorithmic}
\end{algorithm}

The intuitive reason why this does not directly work for non-monotone
functions (\ie $A$ is not
a solution to \tp (Def.~\ref{def:tp})) is: 
if one of the elements added fails the marginal gain
condition, it may do so arbitrarily badly and have a large
negative marginal gain.
Moreover, one cannot
simply exclude such elements from consideration, because they
are needed to ensure that the filtering step at the next iteration will
discard a large enough fraction of elements.
\red{Deleting such elements requires recalculating the marginal gains 
with respect to the updated sets, 
which increases the number of adaptive rounds required in each iteration by a factor of \revtwo{$\oh{k}$}.}
Our solution is to
keep these elements in the set $A$ which is used for filtering
\red{and responsible for Property (2) of \tp (Def.~\ref{def:tp})},
but only include those elements
with a nonnegative marginal gain in the candidate solution set
$A'$\red{, which is responsible for Property (1) of \tp (Def.~\ref{def:tp})}. 
The membership of $A'$ is known since the gain of every element
was computed in parallel. Moreover, $|A'| \ge (1 - \epsi)|A|$
gives the needed relationship on the average marginal gain of each element of $A'$.
\red{Due to submodularity,
the objective value \revtwo{does not decrease} when we 
exclude elements with negative marginal gains.}

\red{\textbf{Discussion of $\delta$.}
Different from other threshold sampling algorithms,
\threseq incorporates an additional input parameter, $\delta$.
This parameter reflects the number of iterations
in the outer for loop,
or specifically the adaptive rounds achieved by the algorithm. 
As the algorithm progresses and more elements are added to the solution set, the size of A  increases while the size of V decreases.
Then, the algorithm stops successfully once $|A| = k$ or $|V| = 0$.
The more iterations, the more likely it is to succeed.
\revtwo{Intuitively, the higher $\delta$ is, the lower is the probability of \threseq choosing a subset that improves on costs and satisfies the constraint.
As stated in Theorem~\ref{thm:threshold} Property (1), \threseq succeeds with a probability greater than $1-\delta/n$.}
For downstream approximation algorithms that \revtwo{use} \threseq as a subroutine
with a specific $\delta$ value,
the more calls made to \threseq,
the lower success probability it achieves.
The adoption of $\delta$ makes such probability manageable.}


\subsection{Theoretical Guarantees}
\begin{theorem} \label{thm:threshold}
Let $(f,k)$ be an instance of \sm . For any constant $\epsi$, 
the algorithm \threseq outputs $A'\subseteq A \subseteq \mathcal{N}$ such that the following properties hold:
\revtwo{\begin{itemize}
\item[1)] The algorithm succeeds with probability at least $1 - \delta/n$.
\item[2)] There are $\oh{n/\epsi}$ oracle queries in expectation and $\oh{\log (n/\delta)/\epsi}$ adaptive rounds.
\item[3)] It holds that $f(A') \ge (1-\epsi)\tau |A|$.
If $|A| < k$, then $\marge{x}{A} < \tau$ for all $x\in \mathcal{N}$.
\item[4)] It also holds that $f(A')\ge f(A)$ and $|A'|\ge (1-\epsi)|A|$
\end{itemize}}
\end{theorem}
\red{The performance of \threseq is derived mainly by answering two questions:} 
1) if a constant fraction of elements can be filtered out
at any iteration with a high probability;
2) if the two sets returned solve \tp (Def.~\ref{def:tp}) indirectly.
In Lemma~\ref{lemma:ThresholdFilterSet} below,
it is certified that the number of elements being deleted
in the next iteration monotonically increases from 0 to $|V|$ 
as the size of the selected set increases.
Then, by probability lemma and concentration bounds \red{(in Appendix~\ref{apx:prob})},
Lemma~\ref{lemma:ThresholdProb} answers the first question.
\begin{restatable}{lemma}{ThresholdFilterSet}
	\label{lemma:ThresholdFilterSet}
    \red{Given $V$} after \textbf{random-permutation} on Line~\ref{line:threshold-permute},
    let $S_i=\{x \in V: \marge{x}{A\cup T_i} < \tau\}$.
    It holds that $|S_0|=0$, $|S_{|V|}|=|V|$, and $|S_{i-1}| \le |S_i|$.
\end{restatable}
\begin{restatable}{lemma}{ThresholdProb} 
	\label{lemma:ThresholdProb}
	It holds that $\prob{i^*<\min\{s, t\}} \le 1/2$.
\end{restatable}
Furthermore, with enough iterations, the candidate set $V$ becomes
empty at some point with a high probability.
Also, since the size of the candidate set $|V|$ exponentially decreases,
intuitively, the total queries is linear in expectation. 

A downside of this bifurcated approach is that a downstream algorithm
receives two sets $A, A'$ instead of one from \threseq.
It is obvious that the second property of \tp \red{(Def.~\ref{def:tp})} holds naturally with set $A$.
\red{Lemma~\ref{lemma:ThresholdGood} below shows how we can relate set $A'$ with set $A$.
Therefore,} by discarding the elements with negative gains in $A$,
the gains of the rest elements\red{, denoted by $A'$,} 
increase and follow the first property of \tp \red{(Def.~\ref{def:tp})}.
\begin{restatable}{lemma}{ThresholdGood}
	\label{lemma:ThresholdGood}
	Say an element added to the solution set is good if its gain is greater than $\tau$.
	\red{Suppose that Algorithm~\ref{alg:threshold} terminates successfully.}
	$A$ and $A'$ returned by Algorithm~\ref{alg:threshold} hold the following properties: 
	\revtwo{\begin{itemize}
	\item[1)] There are at least $(1-\epsi)$-fraction of $A$ that is good.
	\item[2)] A good element in $A$ is always a good element in $A'$. 
	\item[3)] And, any element in $A'$ has non-negative marginal gain when added.
	\end{itemize}}
\end{restatable}

\red{The proofs of the lemmas above can be found in Appendix~\ref{apx:threseq}.}
Now, we provide the proof concerning the performance of \threseq.



\begin{proof}[Proof of Success Probability (Property 1)]
	The algorithm succeeds if $|V|=0$ or $|A|=k$ at termination.
	If we can filter out a constant fraction of $V$ or select
	a subset with $k-|A|$ elements at any iteration with a 
	constant probability, then, with enough iterations,
	the algorithm successfully terminates with a high probability.

	From Lemma~\ref{lemma:ThresholdFilterSet},
	there exists a point $t$ such that
	$t = \min \{i: |S_i| \ge \epsi |V|/2\}$,
	where the next iteration filters out more than
	$\epsi/2$-fraction of elements if $i^* \ge t$.
	Intuitively, when $i \le t$, there is a \red{constant}
	probability that
	the \red{fraction} of \textbf{true}s in $B[1:i]$
	exceeds $1-\epsi$.
	\red{According to Lemma~\ref{lemma:ThresholdFilterSet}, 
	Lemma~\ref{lemma:ThresholdProb} is provided to give the probability
	that whether $|A| = k$ or $\epsi/2-$fraction of $V$ are filtered out
	at the next iteration.}

	\red{For the purposes of the analysis, consider
		a version of the algorithm that does not break on Line~\ref{line:tsreturn} when $|V| = 0$.}
	If so, \red{in subsequent iterations following $|V| = 0$},
	it is always the case that $s = 0$ and $T_{i^*}=\emptyset$.
	Lemma~\ref{lemma:ThresholdProb} still holds in this case.
	\red{As a result, the algorithm returns the same solution set as
	the original one.}

	\red{When the algorithm fails to terminate, at each iteration,
	    it always holds that $i^* < s$;
	    and there are no more than $m=\lceil\log_{1-\epsi/2}(1/n) \rceil$ 
	    iterations that $i^* \ge t$.
	    Therefore, there are no more than $m$ iterations that
	    $i^* \ge \min\{s,t\}$.
	    Otherwise, with more than $m$ iterations that
	    $i^* \ge \min\{s,t\}$, 
	    if there is an iteration that
	    $s \le t$, the algorithm terminates with $|A| = k$.
	    Otherwise, with more than $m$ iterations that 
	    $i^* \ge t$, the algorithm terminates with $|V| = 0$.
	    Define a \textit{successful iteration} as an iteration that
	    $i^* \ge \min\{s,t\}$, which means it successfully filters out 
	    $\epsi/2$-fraction of $V$ or the algorithm stops here.
	    Let $X$ be the number of successes in the $\ell$ iterations.
	    Then, $X$ can be regarded as a sum of dependent Bernoulli trails,
	    where the success probability is larger than 1/2 
	    from Lemma~\ref{lemma:ThresholdProb}.
	    Let $Y$ be a sum of independent Bernoulli trials,
	    where the success probability is equal to 1/2.
	    Then, the probability of failure can be bounded as follows,
	    \begin{align*} 
	        \prob{\text{failure}} &\le \prob{X \le m} 
	        \overset{(a)}{\le} \prob{Y \le m} 
	        \le \prob{Y \le 2\log(n)/\epsi}\\ 
	        &\overset{(b)}{\le} e^{-\left(\frac{\log(n)}{2\log(n)+
	        \epsi \log\left(\frac{n}{\delta}\right)}-1\right)^2
	        \cdot \left(\frac{2}{\epsi}\log(n)+\log\left(\frac{n}{\delta}\right)\right)}
	        \le \frac{\delta}{n},
	    \end{align*}
	    where Inequality (a) follows from Lemma~\ref{lemma:indep},
	    and Inequality (b) follows from Lemma~\ref{lemma:chernoff}.}
	\end{proof}

	\begin{proof}[Proof of Adaptivity and Query Complexity (Property 2)]
	In Alg.~\ref{alg:threshold}, the oracle queries occur on 
	Line~\ref{line:threshold-filtering} and~\ref{line:threshold-if}.
	Since filtering and inner \textbf{for} loop can be done in parallel, 
	there are constant adaptive rounds in an iteration.
	Therefore, the adaptivity is $\oh{\ell} = \oh{\log(n/\delta)/\epsi}$.
	
	As for the query complexity, 
	let $V_j$ be the set $V$ after filtering on Line~\ref{line:threshold-filtering}
	in iteration $j$.
	\red{Let $j_{i}$ be the $i$-th successful iterations,
	    $Y_i=j_{i}-j_{i-1}$.
	    By Lemma~\ref{lemma:indep} in Appendix~\ref{apx:prob}, it holds that $\ex{Y_i} \le 2$.
	    For any iteration $j$ that $j_{i-1}+1 \le j \le j_i$,
	    there are $i-1$ successes before it.
	    Thus, it holds that $|V_j| \le n(1-\epsi/2)^{i-1}$.}
	
    \red{At any iteration $j$, there are $|V_{j-1}|+1$ oracle queries on Line~\ref{line:threshold-filtering}.
    As for the inner \textbf{for} loop, there are no more than $|V_j|+1$ oracle queries.
    The expected number of total queries can be bounded as follows:
    \begin{align*}
        \ex{\text{Queries}} &\le \sum_{j=1}^{\ell} \ex{|V_{j-1}|+|V_j|+2}
        \le n+2\ell+\sum_{j=1}^{\ell} 2\ex{|V_j|}\\
        &\le n+2\ell+\sum_{i\ge 1} 2\ex{Y_i\cdot n(1-\epsi/2)^{i-1}}
        \le n+2\ell+4n/\epsi.
    \end{align*}
    Therefore, the total queries are $\oh{n/\epsi}$\red{, where $\epsi \in (0,1)$}.}
	\end{proof}

	\begin{proof}[Proof of Marginal Gains (Property 3 and 4)]
	\red{The algorithm terminates successfully if either 
		$| V | = 0$ or $| A | = k$ during its execution. 
		As proved above, this happens with a probability of as least $1 -\delta/n$. 
		In the proof below, we condition on the event that the algorithm terminates successfully and returns $A$, $A'$.}

	\red{If the algorithm returns $A$ such that $| A | < k$, then it must be the case that the algorithm terminates with $| V | = 0$.}
	So, for any $x \in \mathcal{N}$,
	there exists an iteration $j_{(x)}+1$ such that
	$x$ is filtered out at iteration $j_{(x)}+1$.
	\red{Let $A_{j_{(x)}}$ be A after iteration $j_{(x)}$.}
	Then, due to submodularity \red{and $A_{j_{(x)}}\subseteq A$, 
	it holds that
	\[\marge{x}{A} \le \marge{x}{A_{j_{(x)}}}<\tau.\]}
	\red{Lemma~\ref{lemma:ThresholdGood} applies to any case in which
	the algorithm terminates successfully.
	As a reminder, an added element is considered good 
	if its gain is greater than $\tau$ with respect to the solution
	prior to its inclusion.
	As per Line~\ref{line:threshold-Aprime},
	$A'$ includes all such good elements that are in $A$.
	Based on Property 1 of Lemma~\ref{lemma:ThresholdGood}, 
	it is guaranteed that 
	the number of good elements in $A'$ is more than $(1-\epsi)|A|$.
	Hence, we have
	$|A'|\ge (1-\epsi)|A|$.}

	\red{Furthermore, 
	due to the diminishing returns property of submodular functions,
	removing an element from a set in a sequence will result in 
	non-increasing marginal gains for the remaining elements.}
	For any $x \in A$, let $A_{(x)}$ be a subsequence of $A$ before
	$x$ is added into $A$. Define $A_{(x)}'$ analogously.
	\red{Then, consider any $x \in A$, if $x\not \in A'$,
	it implies that $\marge{x}{A_{(x)}} < 0$;
	if $x \in A'$, it holds that $\marge{x}{A_{(x)}'}\ge \marge{x}{A_{(x)}}$.
	Therefore,
	\[f(A') = \sum_{x \in A'}\marge{x}{A_{(x)}'}
	\ge \sum_{x \in A'}\marge{x}{A_{(x)}}
	+\sum_{x \in A\backslash A'}\marge{x}{A_{(x)}}
	\ge f(A).\]}

	\red{By Property 2 and 3 of Lemma~\ref{lemma:ThresholdGood}, 
	if an element $x$ in $A$ is good,
	it holds that $\marge{x}{A_{(x)}'} \ge \tau$;
	if not, 
	it holds that $\marge{x}{A_{(x)}'} \ge 0$.
	Then,}
	\[f(A')= \sum_{x \in A'}\marge{x}{A_{(x)}'}
	\ge \sum_{x \in A', x \text{ is good}}\marge{x}{A_{(x)}'}
	\ge (1-\epsi)\tau |A|.\]

\end{proof}

\section{The \algOnefullname Algorithm} \label{sec:atg}
\begin{algorithm}[t]
    \caption{The \adapst Algorithm}
    \label{alg:atg}
    \begin{algorithmic}[1]
      \Procedure{AST}{$f, \mathcal{N}, k, \epsi$}
      \State \textbf{Input:} evaluation oracle $f:2^{\mathcal N} \to \reals$, constraint $k$,
      accuracy parameter $\epsi > 0$
      \State Initialize $M \gets \max_{x \in \mathcal N} f(x)$; 
      $c \gets 4 + \alpha$, where $\alpha^{-1}$ is ratio of \unc;
      $\ell \gets \lceil\log_{1 - \epsi}(1/(ck)))\rceil$
      \For{ $i \gets 0$ to $\ell$ in parallel}\label{line:forconcurrent}
      \State $\tau_i \gets M \left( 1 - \epsi \right)^i$
      \State $A_i,A_i' \gets \threseq \left( f, k, \tau_i, \epsi, 1 / 2 \right) $
      \State $B_i,B_i' \gets \threseq \left( f\restriction_{\mathcal N\setminus A_i}, k, \tau_i, \epsi, 1 / 2 \right)$
      \State $A_i'' \gets \unc (A_i)$
      \State $C_i \gets \argmax \{ f(A_i'), f(B_i'), f(A_i'') \}$
      \EndFor
      \State \textbf{return} $C \gets \argmax_i \{ f(C_i) \}$
      \EndProcedure
\end{algorithmic}
\end{algorithm}
In this section, we present the simple algorithm \algOnefullname 
(\atg, Alg.~\ref{alg:atg}) and show it obtains \red{an approximation} ratio of $1/6 - \epsi$ 
with nearly optimal query and adaptivive complexity.
This algorithm relies on running \threseq for a suitably chosen
threshold value. A procedure for
unconstrained maximization is also required.

\textbf{Overview of Algorithm.}
Algorithm \atg works as follows.
First, the \textbf{for} loop 
guesses a value of $\tau$ close to
$\frac{\opt}{(4 + \alpha)k}$, where 
$1/\alpha$ is the ratio of the algorithm used for the unconstrained maximization
subproblem. 
Next, \threseq is called with parameter $\tau$ to yield set $A$ and $A'$;
followed by a second call to \threseq with $f$ restricted to $\mathcal N \setminus A$
to yield set $B$ and $B'$. Next, an unconstrained
maximization is performed with $f$ restricted to $A$ to yield
set $A''$.
Finally, the best of the three candidate sets $A',B',A''$ is returned.

We prove the following theorem concerning the performance of \atg.

\begin{theorem}
  \label{thm:atg}
  Suppose there exists an $(1/\alpha )$-approximation for
  \unc with adaptivity \red{$\Phi$} and query complexity
  $\Xi$, and let $\epsi > 0$.
  Then there exists an algorithm for \sm with
  expected approximation ratio $\frac{1}{4+\alpha}-\epsi$ with probability
  at least $1 - 1/n$, expected query complexity
  $\oh{ \log_{1 - \epsi}(1/k) \cdot \left( n / \epsi + \Xi \right) }$,
  and adaptivity $\oh{\log( n) / \epsi + \red{\Phi} }$.
\end{theorem}
\noindent If the algorithm of \citeA{Chen2018b} is used for \unc,
\atg achieves ratio $1/6 - \epsi$
with adaptivive complexity $ \oh{\log( n) / \epsi + \log(1 / \epsi) / \epsi}$ 
and query complexity
$\oh{\log_{1 - \epsi}(1/k) \cdot \left( n / \epsi + n \log^3(1 / \epsi) / \epsi^4 \right)}$.
\revtwo{In the experiment, \unc is implemented to use a random subset which gives an expected $(1/4)-$approximation ratio~\cite{Feige2011}. In this case, \atg achieves ratio $1/8-\epsi$ with adaptivity $\oh{\log( n) / \epsi}$ and query complexity
$\oh{\log_{1 - \epsi}(1/k) n / \epsi}$.}

\textbf{Overview of Proof.}
The proof uses the following strategy: either \threseq finds
a set $A'$ or $B'$ with value \red{approximately} $\tau k$,
which is sufficient to achieve the ratio, or we
have two disjoint sets $A$, $B$ of size less than
$k$, such that for any $x \not \in A \cup B$, $\marge{x}{A} < \tau$
and $\marge{x}{B} < \tau$. 
In this case, for any set $O$, we have
by submodularity \red{and nonnegativity}, $f(O) \le f(O \cap A) + f(O \setminus A )$.
The first term is bounded by the unconstrained maximization, and
the second term is bounded by an application of submodularity and the
fact that the maximum marginal gain of adding an element into $A$ or $B$ is below $\tau$.
The choice of constant $c$ balances the trade-off between the two
cases of the proof. 


\begin{proof}[Proof of Theorem~\ref{thm:atg}]
  Let $(f,k)$ be an instance of \sm, and let $\epsi> 0$.
  Suppose algorithm \atg uses a procedure for \unc with
  expected ratio $1/\alpha$.
  \red{We will show that, with some events that happen with probability of at least $(1 - 1/n)$, 
  the set $C$ returned by
  algorithm $\atg(f,k,\epsi)$ satisfies
  $\ex{f(C)} \ge \left( c^{-1}-\epsi \right) \opt$,} where $\opt$
  is the optimal solution value on the instance $(f,k)$.

  \begin{figure}[ht]\centering
  \includegraphics[width=0.5\textwidth]{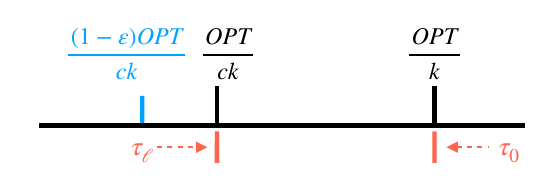}
  \caption{Value of $\tau_0$ and $\tau_\ell$. It is satisfied that
  $\tau_0 \ge \opt / k$ and $\tau_\ell \le \opt / (ck)$.}
  \label{fig:1}
  \end{figure}
  Observe that $\tau_0 = M = \max_{x \in \mathcal N} f(x) \ge \opt / k$ 
  by submodularity of $f$;
  $\tau_\ell = M(1 - \epsi)^\ell \le \opt / (ck)$ 
  since $M \le \opt$. 
  \red{To better explain it, we attach Fig.~\ref{fig:1} above.}
  Because $\tau_i$ decreases by a factor of $1-\epsi$, 
  there exists $i_0$ such that 
  $ \frac{(1 - \epsi)\opt}{ck} \le \tau_{i_0} \le \frac{\opt}{ck}$.
  Let $A,A',B,B',A''$ denote 
  $A_{i_0},A'_{i_0},B_{i_0},B'_{i_0},A''_{i_0}$, respectively.
  For the rest of the proof, we assume that
  the properties of Theorem~\ref{thm:threshold} hold for the calls to \threseq
  with threshold $\tau_{i_0}$, which happens with at least
  probability $1 - 1/n$ by the union bound. 
  
  \textbf{Case $|A| = k$ or $|B| = k$}.
  \red{We suppose that $|A|=k$,
  the proof for the case $|B|  = k$ is directly analogous.}
  By Theorem~\ref{thm:threshold} and the value of $\tau_{i_0}$,
  it holds that,
  \begin{align*}
    f(A') \ge (1-\epsi) \tau_{i_0} |A| \ge 
    \frac{(1-\epsi)^2\opt}{c} \ge (1/c-\epsi)\opt.
  \end{align*}
  Then $f(C) \ge f(A') \ge (1/c - \epsi) \opt$.

  \textbf{Case $|A| < k$ and $|B| < k$}.
  Let $O$ be a set such that $f(O) = \opt$ 
  and $|O| \le k$. 
  Since $|A| < k$, by Theorem~\ref{thm:threshold},
  it holds that for any $x \in \mathcal N$,
  $\marge{x}{A} < \tau_{i_0}$. 
  Similarly, for any $x \in \mathcal N \setminus A$,
  $\marge{x}{B} < \tau_{i_0}$. 
  Hence, by submodularity,
  \begin{equation}
    \numberthis
    \label{ineq:1}
    f(O \cup A) - f(A) \le \sum_{o \in O} \marge{o}{A}
    < k \tau_{i_0} \nonumber \le \opt / c, 
  \end{equation}
  \begin{equation} 
    \numberthis{}
    \label{ineq:2}
    f( (O \setminus A) \cup B ) - f(B) 
    \le \sum_{o \in O \setminus A} \marge{o}{B} < 
    k \tau_{i_0} \nonumber \le \opt / c. 
  \end{equation}
  Next, from (\ref{ineq:1}), (\ref{ineq:2}), submodularity, nonnegativity, 
  Theorem~\ref{thm:threshold}, and the fact that $A \cap B = \emptyset$, 
  it holds that,
  \begin{align} \label{ineq:3}
    f(A') + f(B') &\ge f(A) + f(B) \nonumber \\ 
    &\ge f(O \cup A) + f( (O \setminus A) \cup B ) - 2\opt/c \nonumber \\ 
    &\ge f(O \setminus A) + f( O \cup A \cup B)- 2\opt/c \nonumber \\ 
    &\ge f( O \setminus A )- 2\opt/c.
  \end{align}
Since \unc is an $\alpha$-approximation, we have
\begin{equation} \label{ineq:4}
  \alpha \ex{f(A'')} \ge f( O \cap A ).
\end{equation}
  
From Inequalities (\ref{ineq:3}), (\ref{ineq:4}),
and submodularity,
  we have
  \begin{align*} 
    \opt = f(O) &\le f( O \cap A ) + f( O \setminus A ) \\ 
    &\le \alpha \ex{f(C)} + 2f(C) + 2\opt / c,
  \end{align*}
  from which it follows that
  $\ex{ f(C) } \ge \opt / c.$

  \textbf{Adaptivity and Query Complexities}.
  The adaptivity of \atg is twice the adaptivity of \threseq
  plus the adaptivity of \unc plus a constant. 
  Further, the total query complexity is $\log_{1 - \epsi}(1/(ck))$ times
  the sum of twice the query complexity of \threseq and the query complexity
  of \unc.
\end{proof}


\section{The \algTwofullname Algorithm} \label{sec:latg}
\begin{algorithm}[t]
  \caption{The \adaptg Algorithm}
  \label{alg:latg}
  \begin{algorithmic}[1]
    \Procedure{ATG}{$f, \mathcal{N}, k, \epsi$}
    \Statex \textbf{Input:} evaluation oracle $f:2^{\mathcal N} \to \reals$, constraint $k$,
    accuracy parameter $\epsi > 0$, failure probability $\delta > 0$
    \State Initialize $c \gets 8 / \epsi$, $\epsi' \gets (1 - 1/e)\epsi / 8 $, 
    $\ell=\lceil \log_{1 -  \epsi'}(1/(ck)) \rceil+1$,
    $\delta\gets 1/ (2\ell)$,
    $M \gets \max_{x \in \mathcal N} f(x)$, $A \gets \emptyset$, 
    $A' \gets \emptyset$, $B \gets \emptyset$, $B' \gets \emptyset$
    \For{ $i \gets 1$ to $\ell$}\label{line:for1}
    \State $\tau \gets M \left( 1 -  \epsi' \right)^{i-1}$
    \State $S,S' \gets \threseq (f_A, k - |A|,  \tau,  \epsi', \delta)$ \label{line:thresh1}
    \State $A \gets A \cup S$
    \State $A' \gets A' \cup S'$
    \State \textbf{if} $|A|=k$ \textbf{then} break
    \EndFor
    \For{ $i \gets 1$ to $\ell$}\label{line:for2}
    \State $\tau \gets M \left( 1 -  \epsi' \right)^{i-1}$
    \State $S,S' \gets \threseq ( f_B\restriction_{ \mathcal N \setminus A }, k - |B|, \tau,  \epsi', \delta)$ \label{line:thresh2}
    \State $B \gets B \cup S$
    \State $B' \gets B' \cup S'$
    \State \textbf{if} $|B|=k$ \textbf{then} break
    \EndFor
    \State $A'' \gets \unc (A,  \epsi')$\label{line:unc}
    \State $C \gets \argmax \{ f(A'), f(B'), f(A'') \}$\label{line:chooseC}
    \State \textbf{return} $C$
    \EndProcedure
\end{algorithmic}
\end{algorithm}
In this section, we present the algorithm \algTwofullname 
(\latg, Alg.~\ref{alg:latg}), 
which achieves ratio $\approx 0.193 - \epsi$ 
in nearly optimal query and
adaptivive complexity. 
The price of improving the ratio of the preceding section 
is an extra $\log(k)$ factor in the adaptivity. 

\textbf{Overview of Algorithm.}
Our algorithm (pseudocode in Alg.~\ref{alg:latg}) works
as follows. Each \textbf{for} loop corresponds to a low-adaptivity
greedy procedure using \threseq with descending thresholds. Thus,
the algorithm is structured as two iterated calls to a greedy algorithm,
where the second greedy call is restricted to select elements outside the
auxiliary set $A$ returned by the first. 
Finally, an unconstrained maximization
procedure is used within the first greedily-selected auxiliary set $A$.
Then, the best
of three candidate sets is returned. 
In the pseudocode for \latg, Alg.~\ref{alg:latg}, \threseq is called with
functions of the form $f_S$, which is defined
to be the submodular function $f_S( \cdot ) = f( S \cup \cdot )$.

At a high level, our approach is the following:
the \iter framework of \citeA{Gupta2010a}
runs two standard greedy algorithms followed by an
unconstrained maximization, which yields an algorithm
with $\oh{nk}$ query complexity and $\oh{k}$ adaptivity. 
We adopt this framework
but replace the standard greedy algorithm with a
novel greedy approach with low adaptivity and query complexity.
To design this novel greedy approach, we modify
the descending thresholds algorithm of \citeA{Badanidiyuru2014},
which has query complexity $\oh{n \log k}$ but very high adaptivity
of $\Omega(n \log k)$.
We use \threseq to lower the adaptivity of the descending thresholds 
greedy algorithm (see
Appendix~\ref{apx:threshgreedy} for pseudocode and a detailed discussion).

For the resulting algorithm \latg, 
we prove a ratio of 
$0.193 - \epsi$ (Theorem~\ref{thm:latg}), which
improves the $1/6$ ratio for \iter proven in \citeA{Gupta2010a}.
Also, by adopting \threseq proposed in this paper,
the analysis of approximation ratio is simplified.
\red{Thanks to the fact that} the contribution of each element added to the
solution set $A'$ is determined,
at least $(1-\epsi)|A|$ elements in the solution set $A'$ 
have marginal gains which exactly exceed the threshold $\tau$,
while the rest of it have non-negative marginal gains.
Therefore, it is not needed to analyze the marginal gain in expectation anymore.
An exact lower bound is given by the analysis of the two greedy procedures.

A simpler form of our arguments shows that
the improved ratio also holds for the original
\iter of \citeA{Gupta2010a}; this analysis is 
given in Appendix
\ref{apx:iter}. 
We prove the following theorem concerning the performance of \latg.
\begin{theorem} \label{thm:latg}
  Suppose there exists an $(1/\alpha)$-approximation for
  \unc with adaptivity \red{$\Phi$} and query complexity
  $\Xi$, and let $\epsi> 0$.
  Then the algorithm \algTwofullname  for \sm
  has expected approximation ratio
  $\atgRatio$ 
  with probability at
  least $(1 - 1/n)$,
  adaptivive complexity of
  $\oh{\log_{1 - \epsi} (1/k) \log( n) / \epsi + \red{\Phi}}$ 
  and expected query complexity
  of
  $\oh{\log_{1 - \epsi}(1/k) \cdot \left( n / \epsi \right) + \Xi}$.
\end{theorem}
  If the algorithm of \citeA{Chen2018b} is used for \unc,
\latg achieves approximation
  ratio $\approx 0.193 - \epsi$
  with adaptivive complexity
  $\oh{\log(n)\log( k )}$
  and query complexity
  $\oh{n \log(k) }$, wherein the $\epsi$
  dependence has been suppressed.
  \revtwo{Same as the \atg algorithm, \unc is implemented to use a random subset, and \latg achieves approximation
  ratio $\approx 0.139 - \epsi$
  with adaptivive complexity
  $\oh{\log(n)\log( k )}$
  and query complexity
  $\oh{n \log(k) }$.}
\begin{proof}[Proof of Theorem~\ref{thm:latg}]
  In this proof, we
  assume that the guarantees 
  of Theorem~\ref{thm:threshold} hold
  for each call to \threseq made by
  \latg; this occurs with probability at least
  $(1 - 1/n)$ by the union bound and 
  the choice of $\delta$.
  
\textbf{Overview of Proof.}
For the proof, a substantial amount of
machinery is necessary to lower bound
the marginal gain. The necessary notations
are made first; then, in Lemmas~\ref{lemm:A} --~\ref{lemm:B},
we formulate the necessary lower bounds on the marginal gains
for the first and second greedy procedures. For each respective
greedy procedure, 
this is accomplished by considering the good elements in the selected set
returned by \threseq, or the dummy element if the size of 
selected set is limited. 
This allows us to formulate a recurrence
on the sum of the marginal gains (Lemma~\ref{lemma:three}).
Finally, the recurrence allows us to proceed similarly to our proof
in Appendix~\ref{apx:iter} after a careful analysis of the error
introduced (Lemma~\ref{lemm:seven} in Appendix~\ref{apx:latg}).
\textbf{Notations.}
  Followed by the notations in the pseudocode of Alg.~\ref{alg:latg},
  $A$ and $A'$ are returned by the first greedy procedure,
  while $B$ and $B'$ are returned by the second one.
  \red{Let $A_i$ be the set $A$ after iteration $i$,
  $a_j'$ be the $j$-th element in $A'$,
  and $i(j)$ be the iteration that returns $a_j'$.
  If $j > |A'|$, let $a_j'$ be a dummy element,
  and $i(j) = \ell+1$.
  Furthermore, define $\mathcal{A}_j' = \{a_1', \ldots, a_j'\}$.}
  Then, we define $B_{i(j)}$ and $\mathcal{B}_j'$ analogously.
\begin{restatable}{lemma}{LemmaA}
  \label{lemm:A}
  For $1 \le j \le k$, there are at least $\lceil(1-\epsi')k\rceil$ \red{inequalities where}
    $$f(\mathcal{A}_j')-f(\mathcal{A}_{j-1}')+\frac{M}{ck}\ge 
    \frac{1-\epsi'}{k}\left(f(O\cup A_{i(j)-1})-f(\mathcal{A}_{j-1}')\right).$$
    And for any $j$,
    $$f(\mathcal{A}_j')\ge f(\mathcal{A}_{j-1}').$$
\end{restatable}
The proof of the above lemma can be found in Appendix~\ref{apx:latg}.
Following the notations and the proof of Lemma~\ref{lemm:A},
we can get an analogous result
for the gain of $\mathcal B'$ as follows.

\begin{restatable}{lemma}{LemmaB} 
  \label{lemm:B}
  For $1 \le j \le k$, there are at least $\lceil(1-\epsi')k\rceil$ of $j$ such that
    $$f(\mathcal{B}_j')-f(\mathcal{B}_{j-1}')+\frac{M}{ck}\ge 
    \frac{1-\epsi'}{k}\left(f((O\backslash A)\cup B_{i(j)-1})-f(\mathcal{B}_{j-1}')\right).$$
    And for any $j$,
    $$f(\mathcal{B}_j')\ge f(\mathcal{B}_{j-1}').$$
\end{restatable}

The next lemma proved in Appendix~\ref{apx:latg} establishes the main recurrence.
\begin{restatable}{lemma}{LemmaC} 
  \label{lemma:three}
  Let $\Gamma_u=f(\mathcal{A}_{j(u)}')+f(\mathcal{B}_{j(u)}')$, 
    where $j(u)$ is the $u$-th $j$ which satisfies
    Lemma~\ref{lemm:A} or Lemma~\ref{lemm:B}.
    Then, there are at least $\lceil(1-\epsi')k\rceil$ of $u$ follow that
    $$f(O\backslash A)-\Gamma_u-\frac{2M}{c(1-\epsi')}\le 
    \left(1-\frac{1-\epsi'}{k}\right)\left(f(O\backslash A)-
    \Gamma_{u-1}-\frac{2M}{c(1-\epsi')}\right).$$
  \end{restatable}

Lemma~\ref{lemma:three} yields a recurrence of the form 
$\left(b-u_{i+1}\right) \le a\left(b-u_{i}\right)$, $u_0 = 0$,
and has the solution $u_i \ge b(1 - a^i)$.
Consequently, we have 
\begin{align} \label{eq:part}
f(A') + f(B') &\ge \left(1-\left(1-\frac{1-\epsi'}{k}\right)^{(1-\epsi')k}\right)
\left(f(O \backslash A) - \frac{2M}{c(1-\epsi')}\right) \nonumber \\
& \ge \left(1-e^{-(1-\epsi')^2}\right)
\left(f(O \backslash A) - \frac{2M}{c(1-\epsi')}\right)
\end{align}
Let $\beta = 1 - e^{-(1 - \epsi')^2}$. From
the choice of $C$ on line~\ref{line:chooseC}, we have $2f(C) \ge f(A') + f(B')$ and
so from (\ref{eq:part}), we have 
\begin{align} \label{eq:part3}
f(O \setminus A) &\le \frac{2}{\beta} f(C) + \frac{2M}{c(1 - \epsi')} \nonumber \\ 
&\le \frac{2}{\beta} f(C) + \frac{2f(O)}{c(1 - \epsi')}.
\end{align}

Since an $(1/\alpha)$-approximation is used 
for \unc, for any $A$, $f(O \cap A) / \alpha \le \ex{ f(A'') | A }$; therefore, 
\begin{equation} \label{eq:int}
\red{\ex{f( O \cap A )} \le \alpha \ex{f(C)}}.
\end{equation}

For any set $A$, $f(O) \le f(O \cap A) + f(O \setminus A)$
by submodularity and nonnegativity. 
Therefore, by Inequalities~\ref{eq:part3} and~\ref{eq:int},
\begin{align*}
f(O) &\le \ex{f(O \cap A) + f(O \setminus A)}\\
&\le \red{\alpha \ex{f(C)}+ \frac{2}{\beta} \ex{f(C)} + \frac{2f(O)}{c(1 - \epsi')}}.
\end{align*}
Therefore, \red{we have from Lemma~\ref{lemm:seven} in Appendix~\ref{apx:latg},}
$$
  \ex{f(C)}\ge \frac{1-\frac{2}{c(1-\epsi')}}{\alpha+\frac{2}{\beta}}f(O)
  \ge \left(\frac{e-1}{\alpha(e-1)+2e}-\epsi\right)f(O).\qedhere
  $$

\end{proof}

\section{Empirical Evaluation} \label{sec:exp}
In this section, we evaluate our algorithm in comparison
with the state-of-the-art parallelizable algorithms: 
\anm of \citeA{Fahrbach2018a}, the algorithm of \citeA{Ene2020},
\red{and two versions of \park in~\citeA{amanatidis2021submodular}
(\parkone represents the one without binary search
and \parktwo is the one with binary search)}.
Also, we compare four versions of our algorithms
with different threshold procedures: \thresam of \citeA{Fahrbach2018}, 
two versions of threshold sampling algorithms of \citeA{amanatidis2021submodular},
and \threseq proposed in this paper.
Our results are summarized as follows.
~\footnote{Our code is available at 
\textit{https://gitlab.com/luciacyx/nm-adaptive-code.git}.}
\begin{itemize}
\item Our algorithm \latg obtains the best objective
value of any of the parallelizable algorithms;
obtaining an improvement of up to 19\% over the next algorithm,
our \atg. Both \citeA{Fahrbach2018a} and \citeA{Ene2020} exhibit
a large loss of objective value at both small and large $k$ values.
\item Both our algorithm \atg, \red{\parkone,} and \anm use a very small number
of adaptive rounds. Both \latg and the algorithm of \citeA{Ene2020}
use roughly an order of magnitude more adaptive rounds.
\item The algorithm of \citeA{Ene2020} is the most query efficient if access
is provided to an exact oracle 
for the multilinear extension of a submodular function and its 
gradient~\footnote{The definition of the multilinear extension is given in Appendix~\ref{apx:ene}.}.
However, if these oracles must be approximated with the set function, 
their algorithm becomes very inefficient and does not scale 
beyond small instances ($n \le 100$).
\item Our algorithms used fewer queries to the submodular set function 
than the linear-time algorithm \frg
in \citeA{Buchbinder2015a}.
\red{Both versions of \park are the most query inefficient.}
\item Comparing \atg with four threshold sampling algorithms, 
our \threseq proposed in this paper is the most 
query and round efficient without loss of objective values. 
If running \thresam theoretically, 
with a large amount of sampling in \reducedmean,
\red{we empirically establish that the query complexity of algorithms 
using \thresam can be three to four orders of magnitude 
worse than other algorithms over the SMCC instances in our benchmark}.
\end{itemize}


\subsection{\revtwo{Algorithm Setup for \atg and \latg}}
\revtwo{In the pseudocodes for \atg and \latg,
$M$ is used as the upper bound of $\opt/k$ 
which is set to $\max_{x \in \uni} \ff{x}$.
In the experiment, we used a sharper upper bound, the average
of the top $k$ singleton values,
maintaining the analysis of approximation ratio.
Additionally, the $(1/2)-$approximation \unc algorithm is substituted with a random set,
which is a $(1/4)$-approximation by~\citeA{Feige2011}.
Consequently, the obtained approximation ratios for \atg and \latg in the actual experiment are
$1/8-\epsi$ and $0.139-\epsi$, respectively.
}

\subsection{Comparison Algorithms and Other Settings}
In addition to the algorithms discussed in the preceding
paragraphs, we evaluate the following baselines:
the \iter algorithm of \citeA{Gupta2010a},
and the linear-time $(1/e - \epsi)$-approximation algorithm \frg of \citeA{Buchbinder2015a}.
These algorithms are both $\oh{k}$-adaptive, where $k$ is the cardinality constraint.

The algorithm of \citeA{Ene2020} requires access to an oracle
for the multilinear extension and its gradient. In the case of maximum cut,
the multilinear extension and its gradient can be computed in closed form
in time linear in the size of the graph, as described in Appendix~\ref{apx:ene}. 
This fact enables us to evaluate
the algorithm of \citeA{Ene2020} using direct oracle access to the multilinear
extension and its gradient on the maximum cut application. However,
no closed form exists for the multilinear extension of
the revenue maximization objective. In this case, we found (see Appendix~\ref{apx:exp})
that sampling to approximate the multilinear extension is exorbitant in terms
of runtime; hence, we were unable to evaluate \citeA{Ene2020} on revenue maximization.
%

For all algorithms, the accuracy parameter $\epsi$ was set to $0.1$; 
the \revtwo{failure probability parameter} $\delta$ was set to 0.1;
$100$ samples
were used to evaluate expectations for \thresam in \anm (thus, this
algorithm was run as heuristics with no performance guarantee). 
Further, in the algorithms \algTwofullname, \park, and \anm,
we ignored the smaller values of $\epsi$, $\delta$ in each algorithm,
and simply used the input values of $\epsi$ and $\delta$.
For \algTwofullname and \park, 
by using the best solution value found so far
as a lower bound on $\opt$,
we used an early termination condition to check if the threshold
value $\tau < \alpha\opt (1 - \epsi ) / k$, where $\alpha$ is the approximation ratio for each algorithm.
This early termination condition is responsible for
the high variance in total queries. 
We attempted to use the same sharper upper bound \revtwo{on $\opt / k$ as our algorithms} in \anm, but it resulted in significantly worse
objective values, so we simply used the maximum singleton as described in \citeA{Fahrbach2018a}.
\red{\park is generalized by an algorithm that deals with knapsack constraints.
By calling threshold sampling algorithm a large number of times,
\park is able to achieve a constant probability on certain events.
Hence, it is comparatively less efficient than algorithms dealing with cardinality constraints.
For our experiments, we ran only \parkone
and \parktwo on BA and ca-GrQc datasets.}

Randomized algorithms are averaged over $20$ independent repetitions,
and the mean is reported. The standard deviation is indicated by a shaded
region in the plots. Any algorithm that requires a subroutine for $\unc$
is implemented to use a random set, \revtwo{following the setting used for \atg and \latg.}

\subsection{Applications and Datasets} 
\textbf{Maxcut.} The cardinality-constrained maximum cut function is defined as follows.
Given graph $G = (V, E)$, and nonnegative edge weight $w_{ij}$ on each edge
$(i,j) \in E$. For $S \subseteq V$, let
$$f(S) = \sum_{i \in V \setminus S} \sum_{j \in S} w_{ij}.$$
In general, this is a non-monotone, submodular function.
\red{In our implementation, all edges have a weight of 1.}

\textbf{Revmax.} The revenue maximization objective is defined as follows.
Let graph $G = (V, E)$ represent a social network, 
with nonnegative edge weight $w_{ij}$ on each edge
$(i,j) \in E$.
We use the concave graph model introduced by \citeA{Hartline2008}.
In this model, each user $i \in V$ is associated with a non-negative,
concave function $f_i: \reals \to \reals$. The value $v_i(S) = f_i( \sum_{j \in S} w_{ij} )$
encodes how likely the user $i$ is to buy a product if the set $S$ has adopted it.
Then the total revenue for seeding a set $S$ is 
$$f(S) = \sum_{i \in V \setminus S} f_i\left( \sum_{j \in S} w_{ij} \right).$$
This is a non-monotone, submodular function. In our implementation,
each edge weight $w_{ij} \in (0,1)$ is chosen uniformly randomly; further,
$f_i( \cdot ) = ( \cdot )^{\alpha_i}$, where $\alpha_i \in (0,1)$ is chosen
uniformly randomly for each user $i \in V$. 

\textbf{Dataset.} Network topologies from SNAP were used; specifically,
web-Google ($n=875713$, $m = 5105039$), a web graph from
Google, ca-GrQc ($n = 5242, m = 14496$), a collaboration
network from Arxiv General Relativity and Quantum Cosmology,
and ca-Astro ($n = 18772, m = 198110$), a collaboration network
of Arxiv Astro Physics.
In addition, 
a Barabási-Albert random graph was used (BA), with $n = 968$, $m = 5708$.

\begin{figure*}[t] \centering
  \subfigure[Objective, small $k$]{ \label{fig:val-Google}
    \includegraphics[width=0.3\textwidth,height=0.16\textheight]{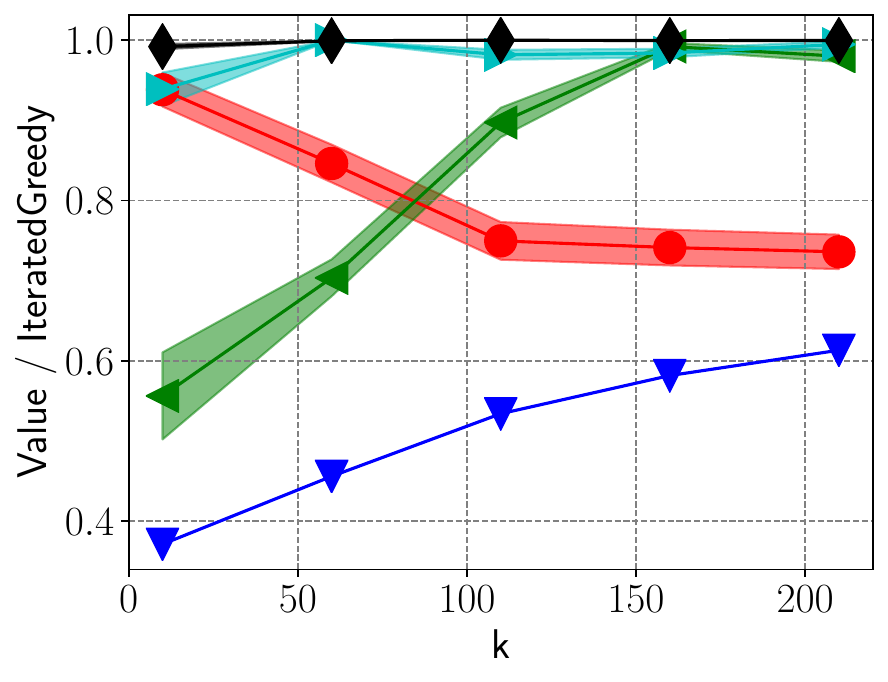}
  }
  \subfigure[Absolute Objective, small $k$]{ \label{fig:legend}
  \includegraphics[width=0.3\textwidth,height=0.16\textheight]{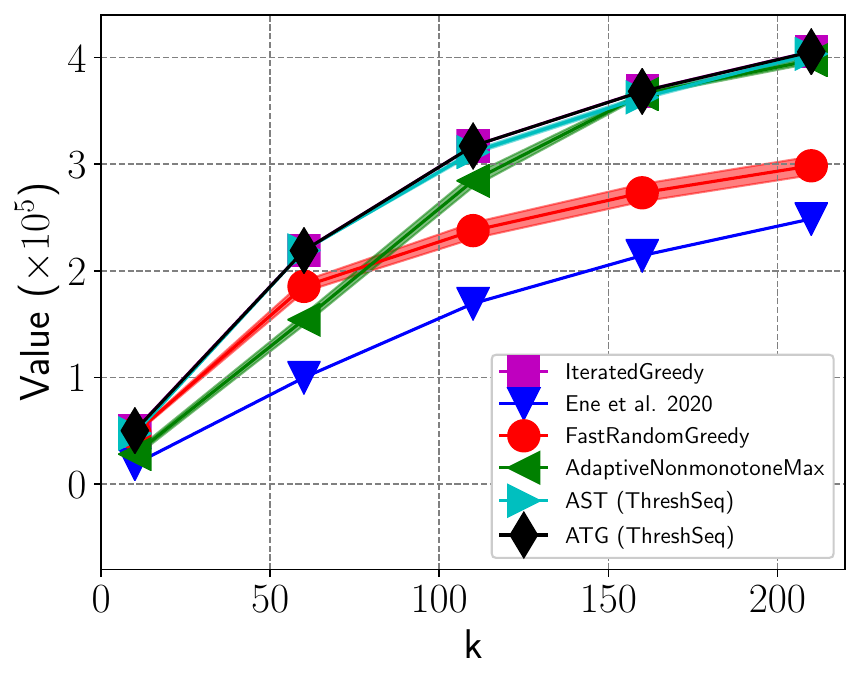}
}
  \subfigure[Rounds, small $k$]{ \label{fig:rounds-Google}
    \includegraphics[width=0.3\textwidth,height=0.16\textheight]{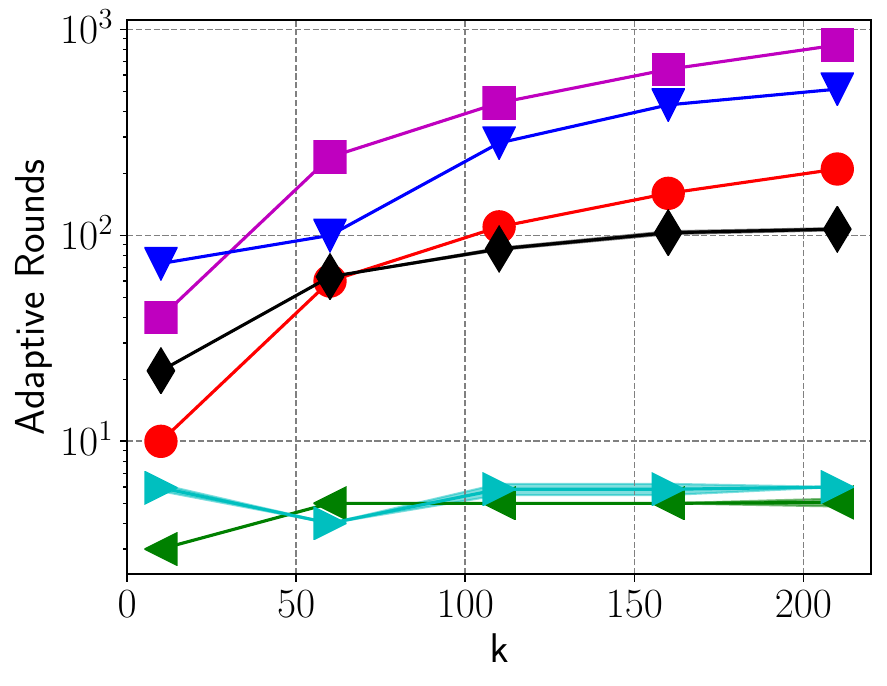}
  }
\subfigure[Objective, large $k$]{ \label{fig:val-Google-largek}
  \includegraphics[width=0.3\textwidth,height=0.16\textheight]{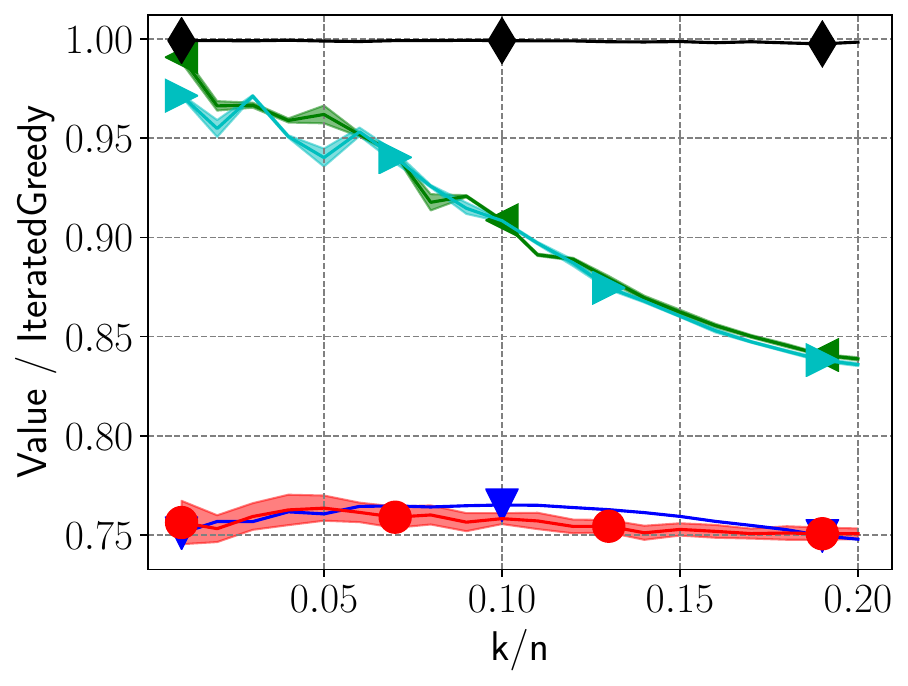}
}
  \subfigure[Queries, large $k$]{ \label{fig:query-Google-largek}
    \includegraphics[width=0.3\textwidth,height=0.16\textheight]{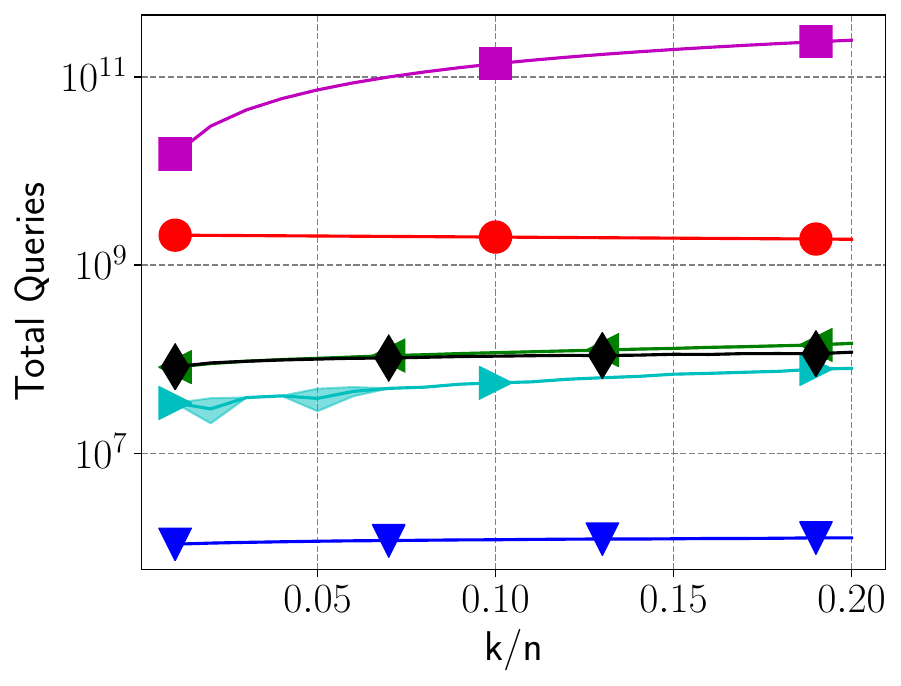}
  }
\subfigure[Rounds, large $k$]{ \label{fig:rounds-Google-largek}
  \includegraphics[width=0.3\textwidth,height=0.16\textheight]{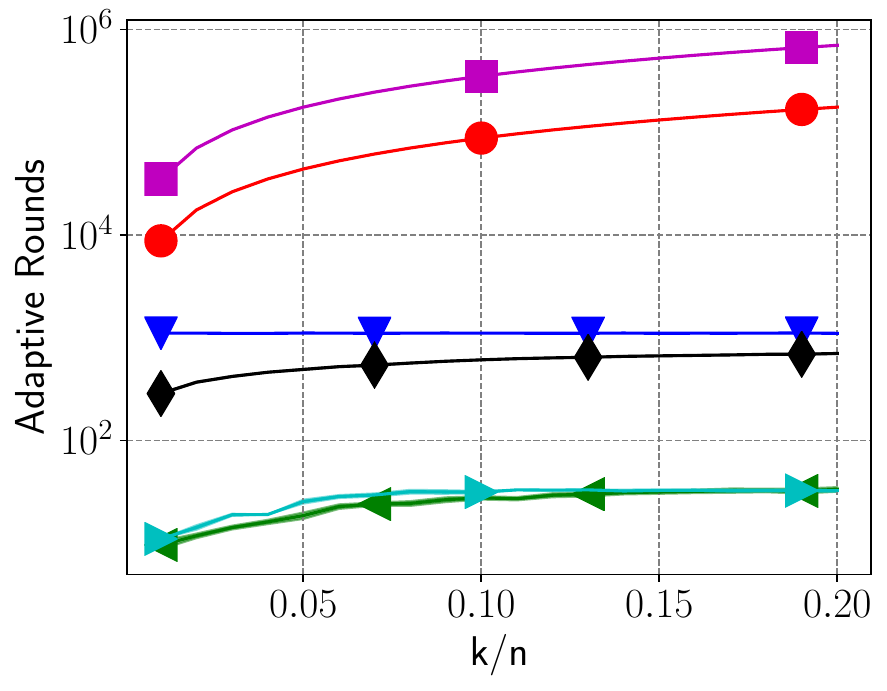}
}
  \caption{Comparison of objective value (normalized by the \iter objective value), 
    total queries, and adaptive rounds on web-Google for the maxcut application for 
    both small and large $k$ values. The large $k$ values are given as a fraction of the number of nodes in the network. The algorithm of \citeA{Ene2020} is run with oracle access to the multilinear extension and its gradient; total queries reported for this algorithm are queries to these oracles, rather than the original set function. The legend in Fig.~\ref{fig:legend} applies to all
    other subfigures.} \label{fig:main}
\end{figure*}
\begin{figure}[ht] \centering 
  \subfigure[Objective, small $k$]{ \label{fig:val-astro}
    \includegraphics[width=0.3\textwidth,height=0.15\textheight]{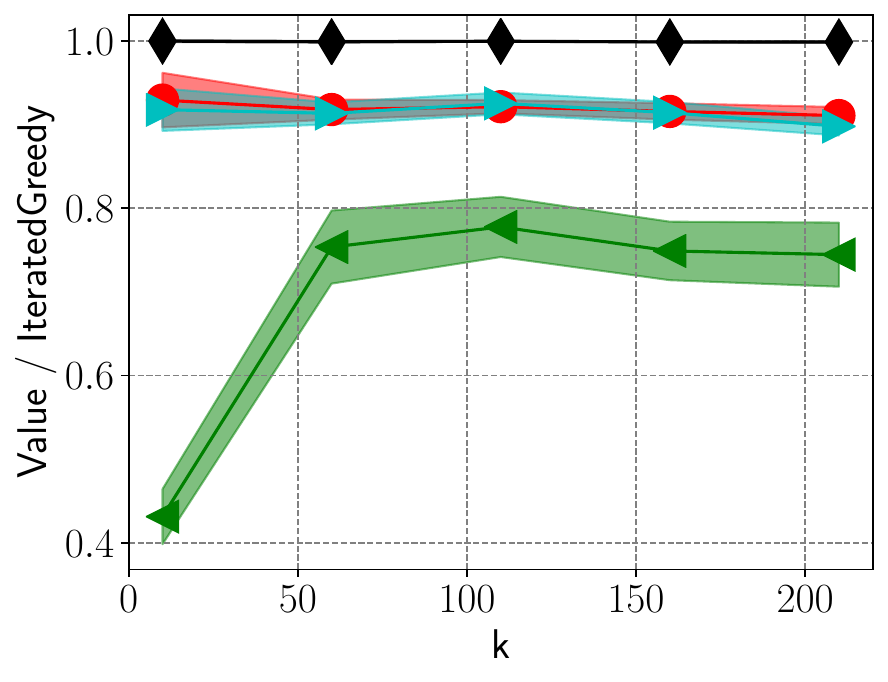}
  }
  \subfigure[Queries, small $k$]{ \label{fig:query-astro}
    \includegraphics[width=0.3\textwidth,height=0.15\textheight]{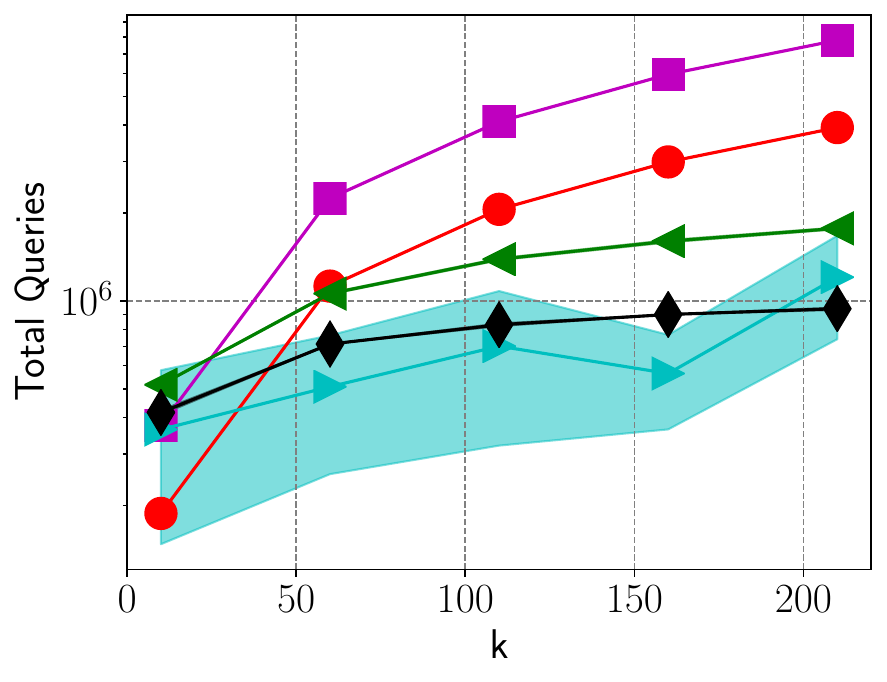}
  }
  \subfigure[Rounds, small $k$]{ \label{fig:rounds-astro}
    \includegraphics[width=0.3\textwidth,height=0.15\textheight]{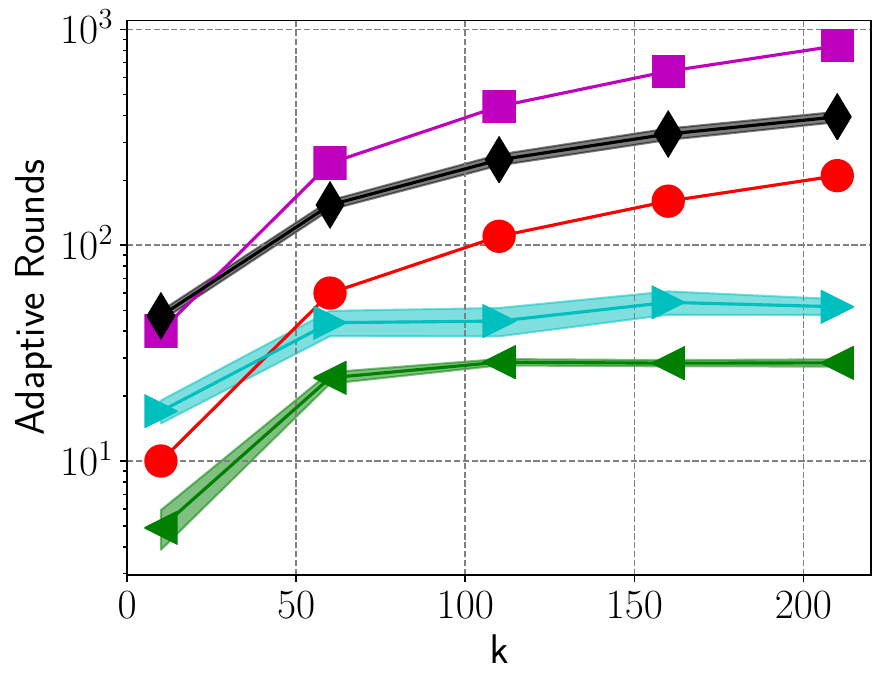}
  }
  \subfigure[Objective, large $k$]{ \label{fig:val-astro-largek}
    \includegraphics[width=0.3\textwidth,height=0.15\textheight]{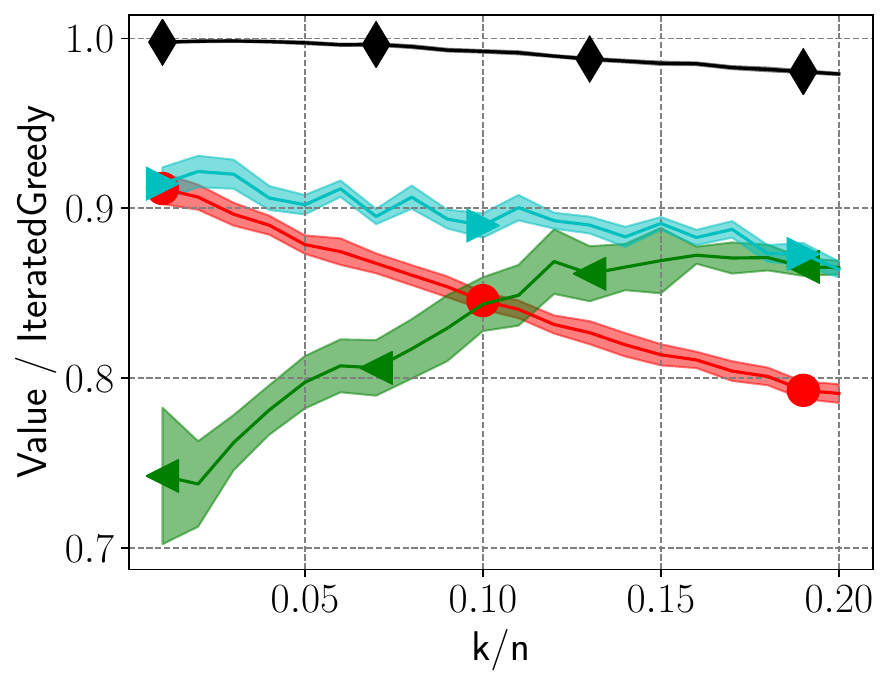}
  }
  \subfigure[Queries, large $k$]{ \label{fig:query-astro-largek}
    \includegraphics[width=0.3\textwidth,height=0.15\textheight]{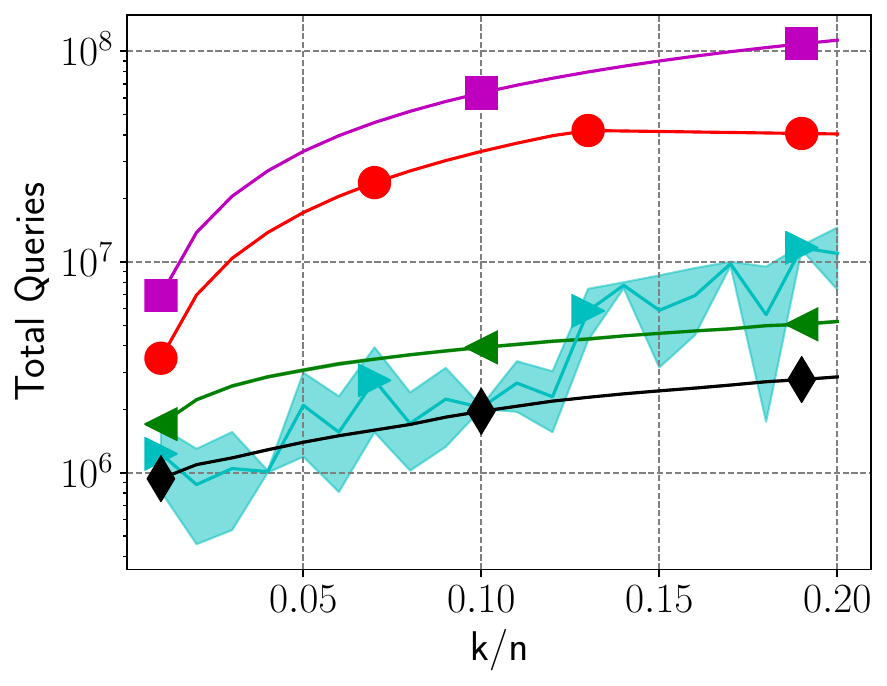}
  }
  \subfigure[Rounds, large $k$]{ \label{fig:rounds-astro-largek}
    \includegraphics[width=0.3\textwidth,height=0.15\textheight]{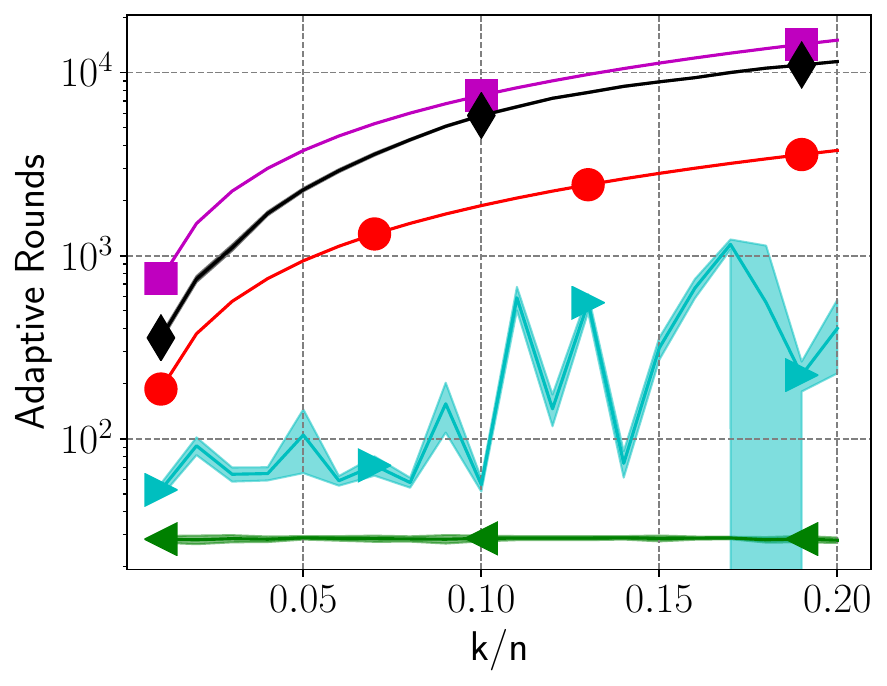}
  }
  \caption{Results for revenue maximization on ca-Astro, for both small and large $k$ values. Large $k$ values are indicated by a fraction of the total number $n$ of nodes. The legends in Fig.~\ref{fig:main} and~\ref{fig:apx-exp-maxcut} apply.} 
  \label{fig:apx-exp-revmax}
\end{figure}
\begin{figure}[ht] \centering
  \subfigure[Objective, BA]{ \label{fig:val-ba}
    \includegraphics[width=0.3\textwidth,height=0.15\textheight]{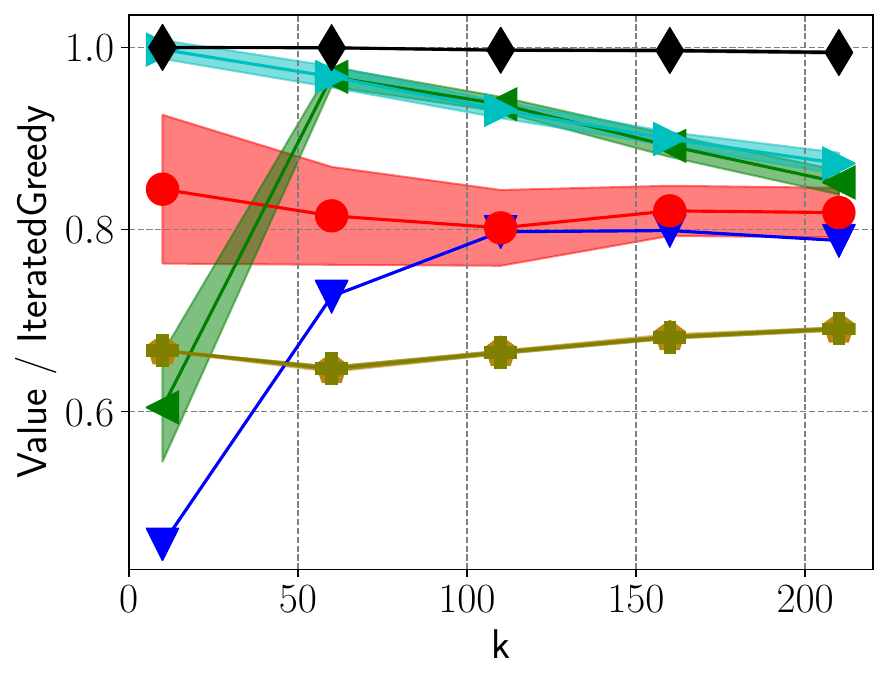}
  }
  \subfigure[Queries, BA]{ \label{fig:query-ba}
    \includegraphics[width=0.3\textwidth,height=0.15\textheight]{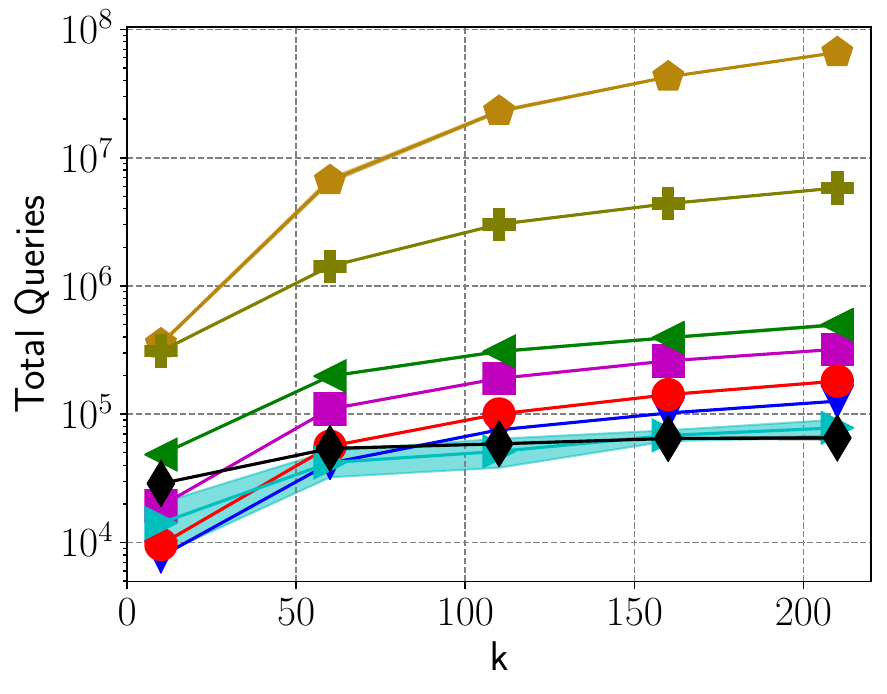}
  }
  \subfigure[Rounds, BA]{ \label{fig:rounds-ba}
    \includegraphics[width=0.3\textwidth,height=0.15\textheight]{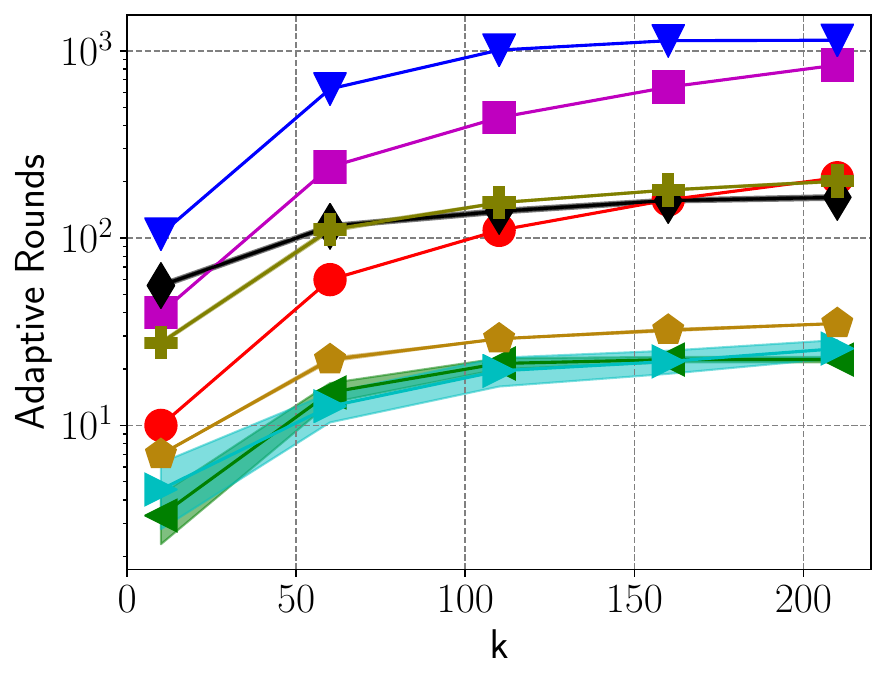}
  }
  \subfigure[Objective, ca-GrQc]{ \label{fig:val-grqc}
    \includegraphics[width=0.3\textwidth,height=0.15\textheight]{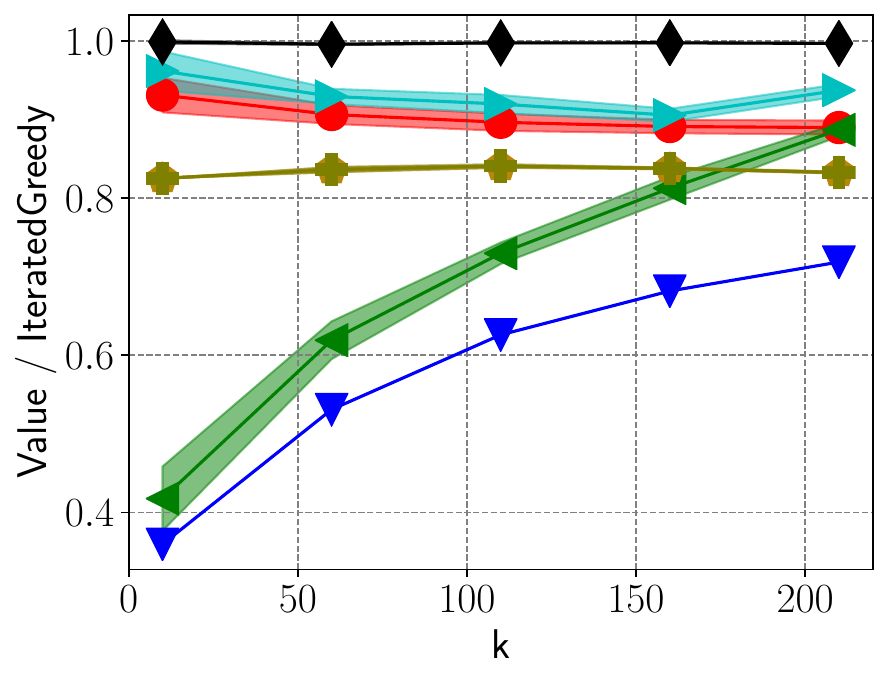}
  }
  \subfigure[Queries, ca-GrQc]{ \label{fig:query-grqc}
    \includegraphics[width=0.3\textwidth,height=0.15\textheight]{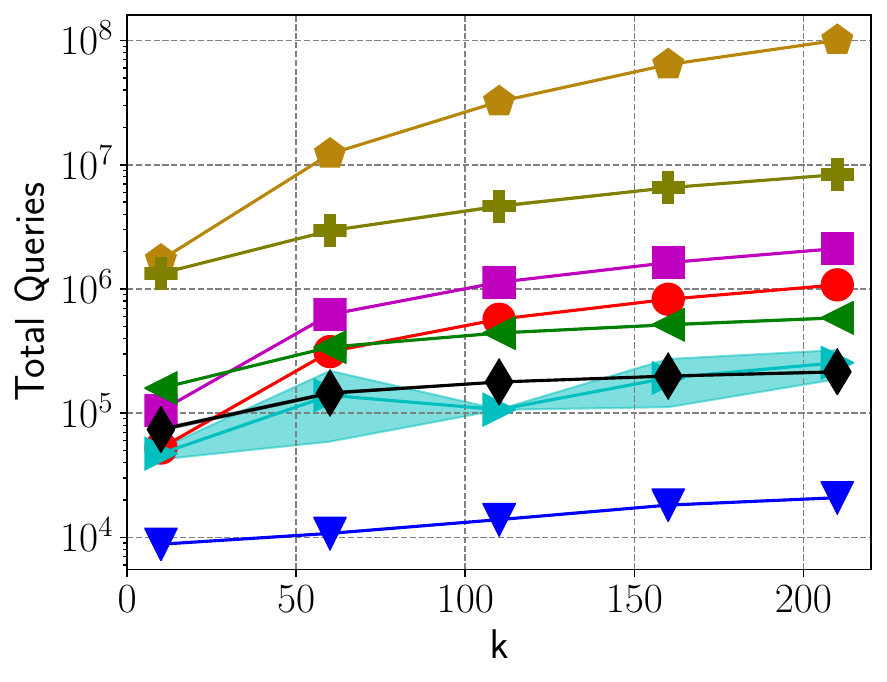}
  }
  \subfigure[Legend]{ \label{fig:rounds-grqc}
    \includegraphics[width=0.3\textwidth,height=0.15\textheight]{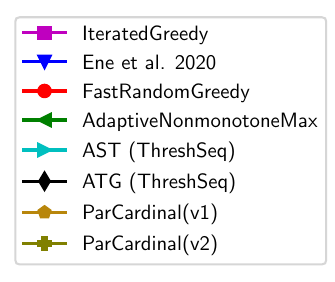}
  }
  \caption{Additional results for maximum cut on BA and ca-GrQc with
  \park algorithms.} 
  \label{fig:apx-exp-maxcut}
\end{figure}

\subsection{Main Results}
In Fig.~\ref{fig:main}, we show representative results for cardinality-constrained
maximum cut on web-Google ($n=875713$)
for both small and large $k$ values. 
Results on other datasets
and revenue maximization are given in Fig.~\ref{fig:apx-exp-maxcut} and~\ref{fig:apx-exp-revmax}.
In addition, results for \citeA{Ene2020} when the multilinear
extension is approximated via sampling are given in Appendix~\ref{apx:exp}.
The algorithms are evaluated by objective value of
solution, total queries made to the oracle, and the number of adaptive rounds (lower is better).
Objective value is normalized by that of \iter. 

In terms of objective value (Figs.~\ref{fig:val-Google} and~\ref{fig:val-Google-largek}), 
our algorithm \latg maintained better than $0.99$ of the \iter value,
while all other algorithms fell below $0.95$ of the \iter value on some instances. 
Our algorithm \atg obtained similar objective value to \anm on larger $k$ values,
but performed much better on small $k$ values. Finally, the algorithm of \citeA{Ene2020}
obtained poor objective value for $k \le 100$ and about $0.75$ of the \iter value
on larger $k$ values. 
It is interesting to observe that the two algorithms with
the best approximation ratio of $1/e$, \citeA{Ene2020} and \frg, returned 
the worst objective values on larger $k$ (Fig.~\ref{fig:val-Google-largek}).
For total queries (Fig.~\ref{fig:query-Google-largek}), 
the most efficient is \citeA{Ene2020}, although
it does not query the set function directly, but
the multilinear extension and its gradient. The most efficient
of the combinatorial algorithms was 
\atg, followed by \latg.
Finally, with respect to the number of adaptive rounds 
(Fig.~\ref{fig:rounds-Google-largek}), the best was \anm, 
closely followed by \atg; the next lowest was \latg, followed by 
\citeA{Ene2020}.

\red{The results in Fig.~\ref{fig:apx-exp-maxcut} and~\ref{fig:apx-exp-revmax}
are qualitatively similar.
Regarding the \park algorithms,
the results in Fig.~\ref{fig:apx-exp-maxcut} demonstrate that
\parktwo is highly parallelizable.
However, despite achieving a 0.172 approximation ratio, 
the objective values of \parkone and \parktwo fell below 0.85 of the \iter.
Because of a constant number of call repetitions to \threseq in \park,
these two algorithms are the most query inefficient
and are roughly two to three orders of magnitude worse
than our algorithms.}
\begin{figure*}[t] \centering
  \subfigure[Objective, BA]{ \label{fig:val-BA-2-ast}
    \includegraphics[width=0.3\textwidth,height=0.16\textheight]{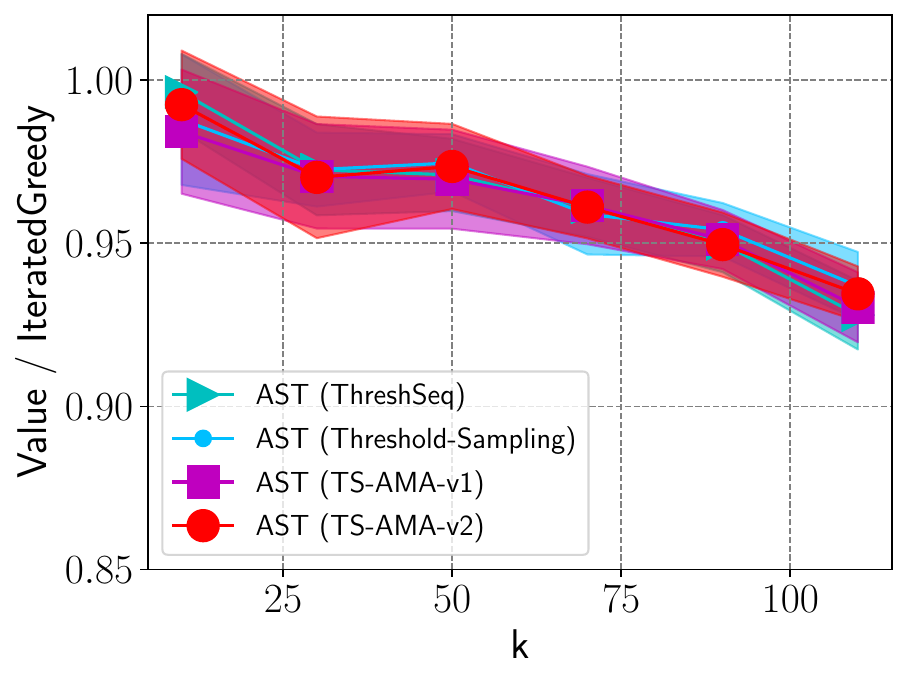}
  }
  \subfigure[Rounds, BA]{ \label{fig:rounds-BA-2-ast}
    \includegraphics[width=0.3\textwidth,height=0.16\textheight]{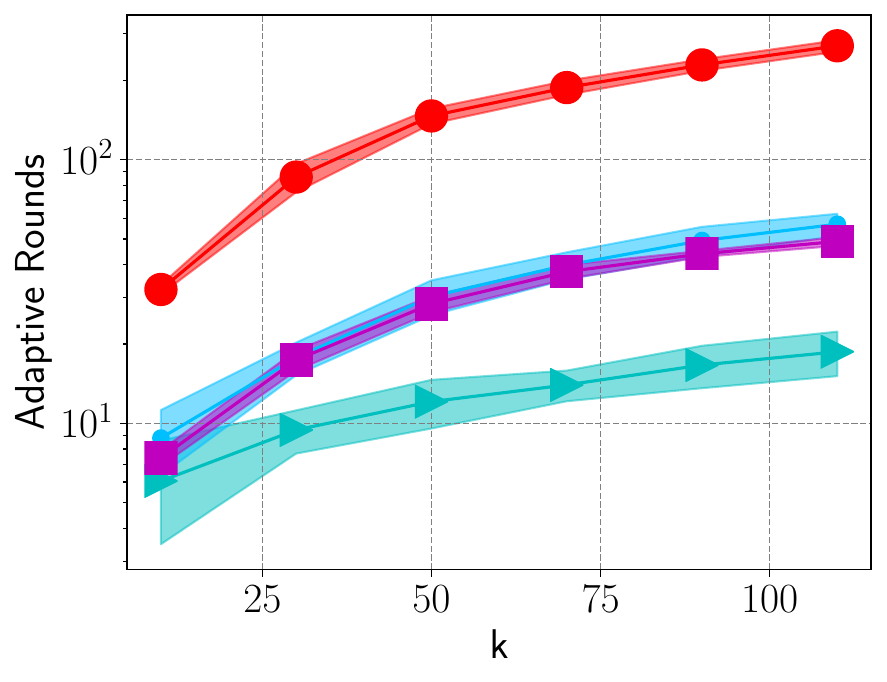}
  }
  \subfigure[Queries, ca-GrQc]{ \label{fig:query-BA-2-ast}
    \includegraphics[width=0.3\textwidth,height=0.16\textheight]{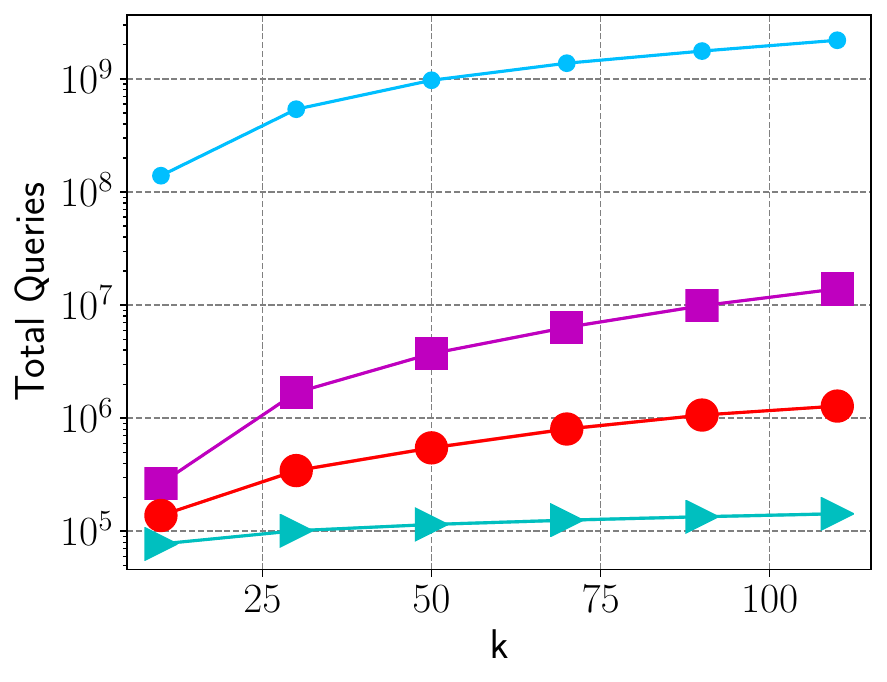}
  }
  \subfigure[Objective, ca-GrQc]{ \label{fig:val-grqc-2-ast}
    \includegraphics[width=0.3\textwidth,height=0.16\textheight]{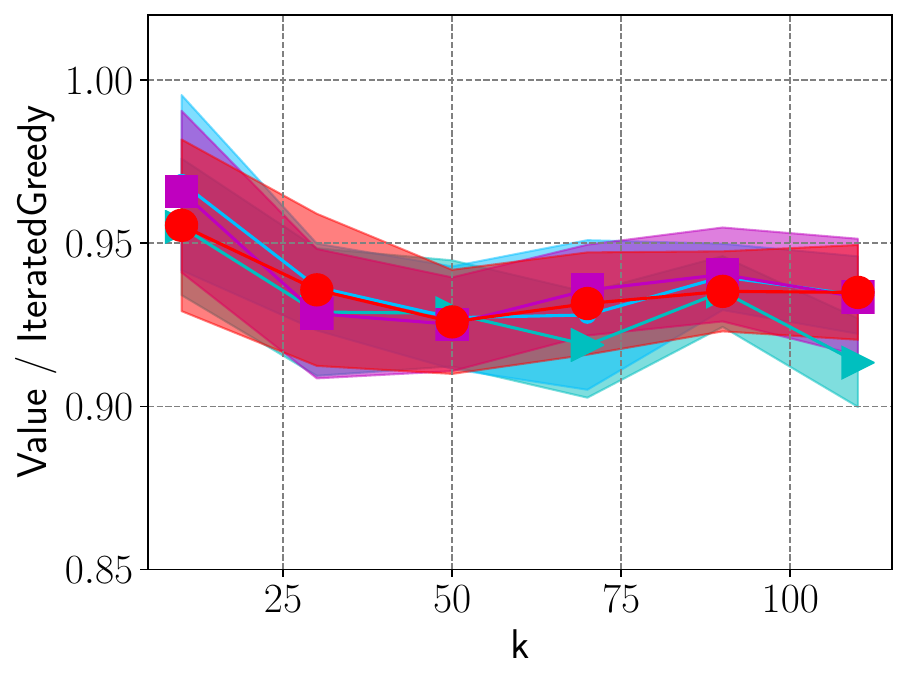}
  }
  \subfigure[Rounds, ca-GrQc]{ \label{fig:rounds-grqc-2-ast}
    \includegraphics[width=0.3\textwidth,height=0.16\textheight]{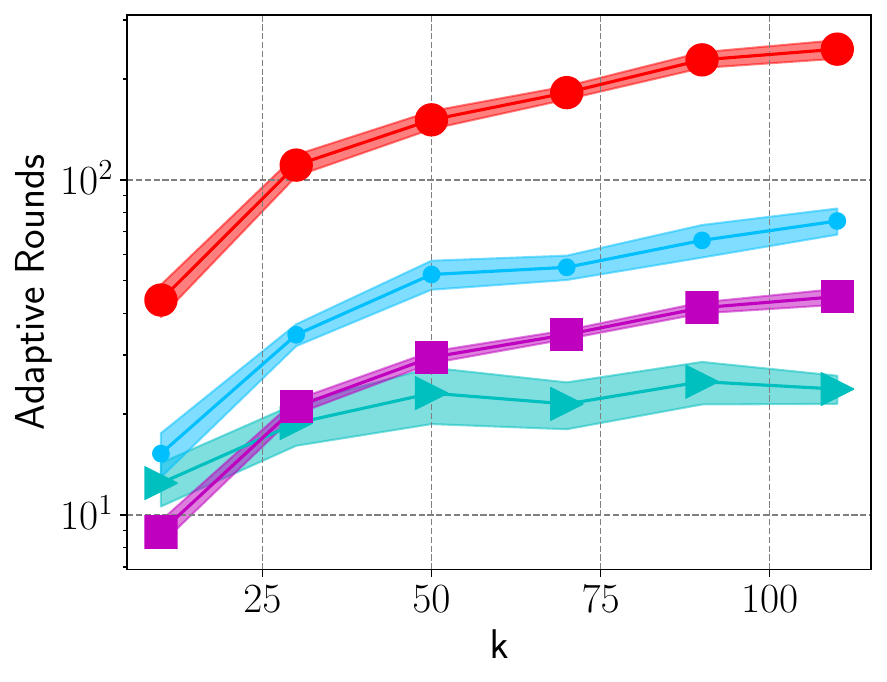}
  }
  \subfigure[Queries, ca-GrQc]{ \label{fig:query-grqc-2-ast}
    \includegraphics[width=0.3\textwidth,height=0.16\textheight]{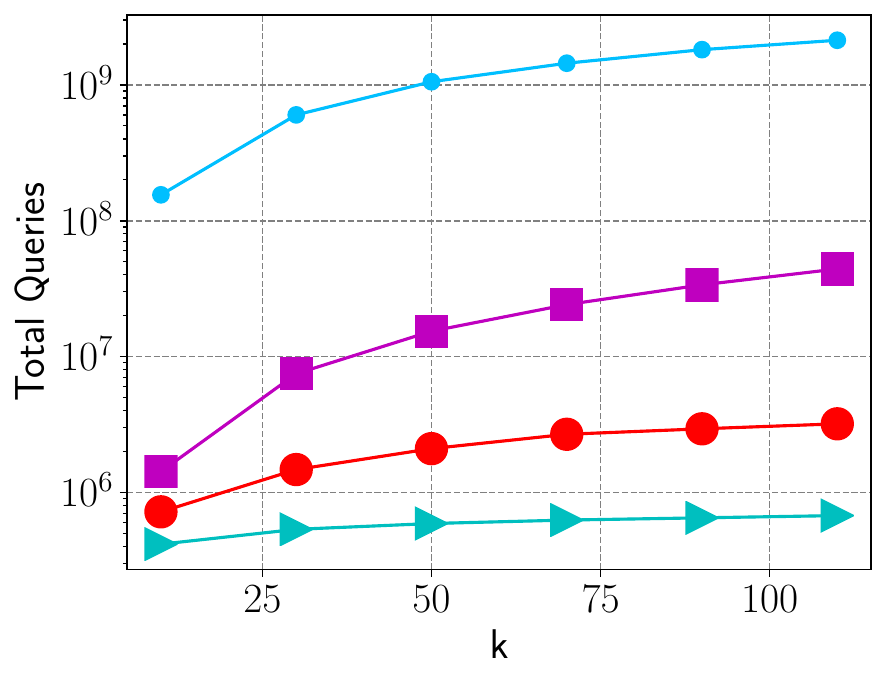}
  }
  \subfigure[Objective, BA]{ \label{fig:val-BA-2-atg}
    \includegraphics[width=0.3\textwidth,height=0.16\textheight]{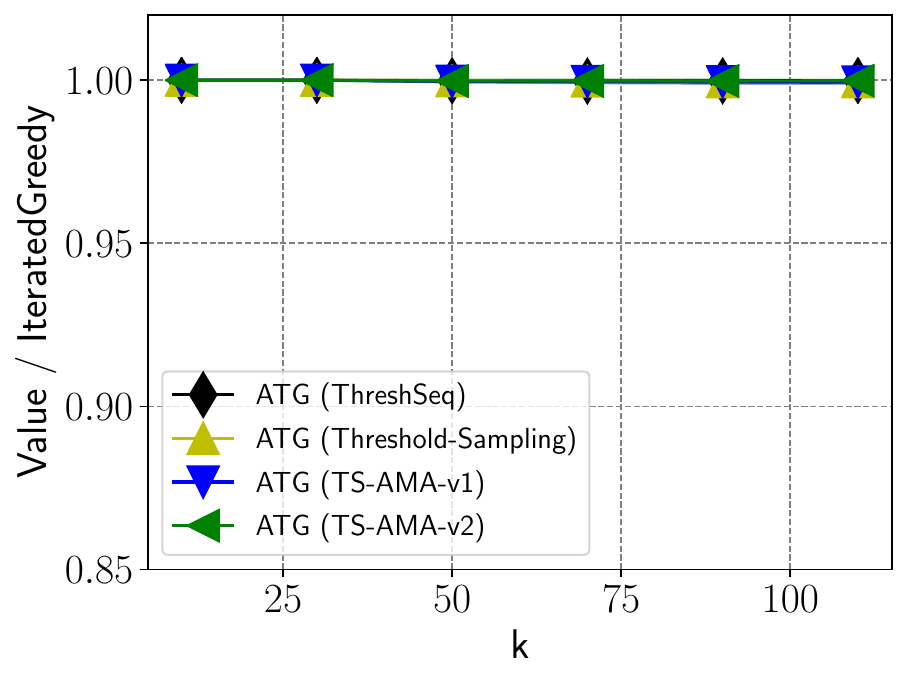}
  }
  \subfigure[Rounds, BA]{ \label{fig:rounds-BA-2-atg}
    \includegraphics[width=0.3\textwidth,height=0.16\textheight]{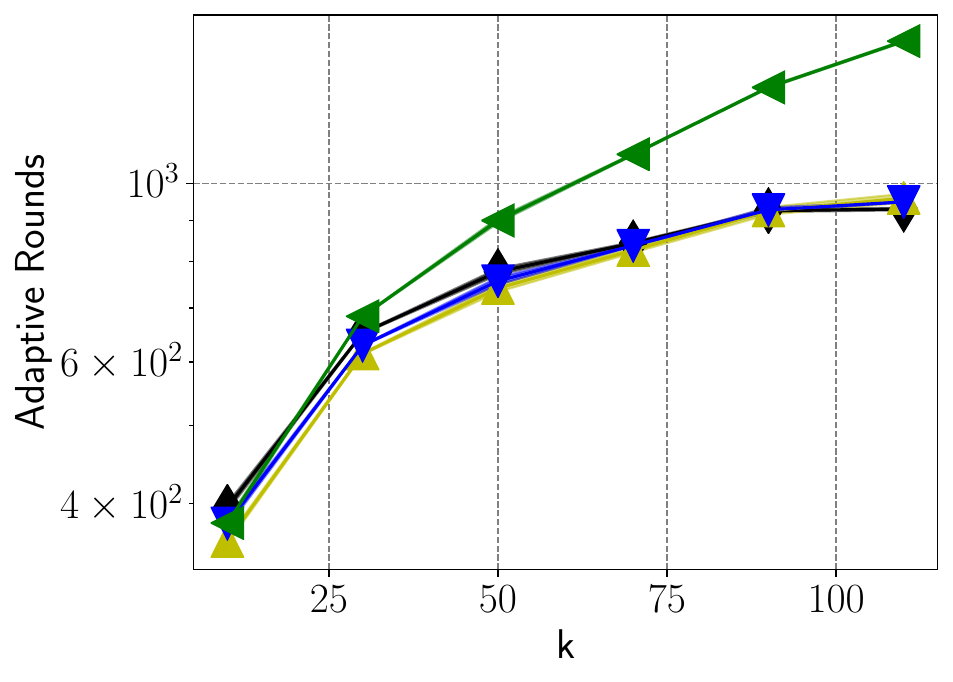}
  }
  \subfigure[Queries, ca-GrQc]{ \label{fig:query-BA-2-atg}
    \includegraphics[width=0.3\textwidth,height=0.16\textheight]{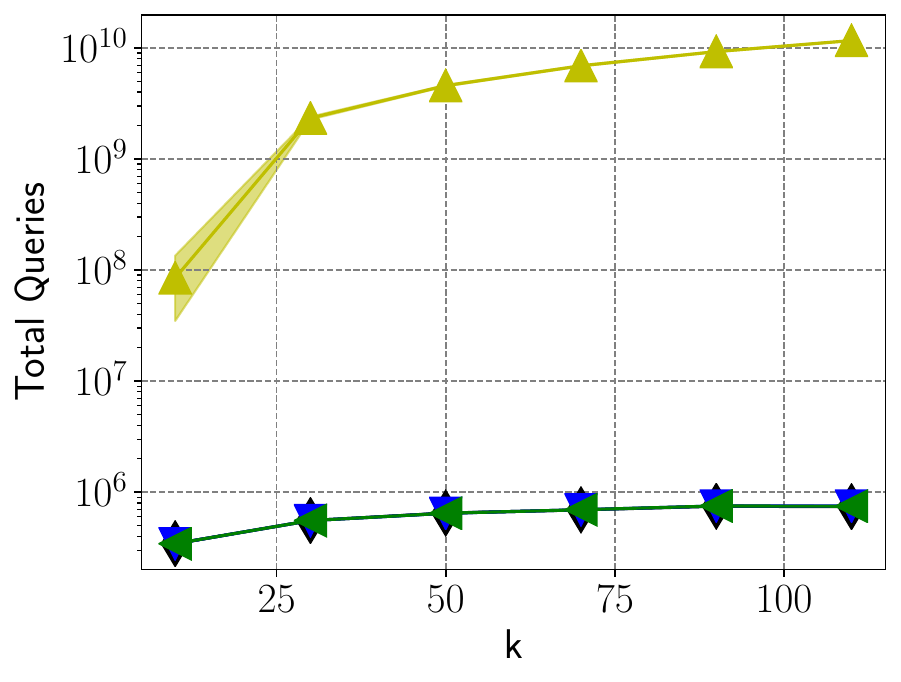}
  }
  \subfigure[Objective, ca-GrQc]{ \label{fig:val-grqc-2-atg}
    \includegraphics[width=0.3\textwidth,height=0.16\textheight]{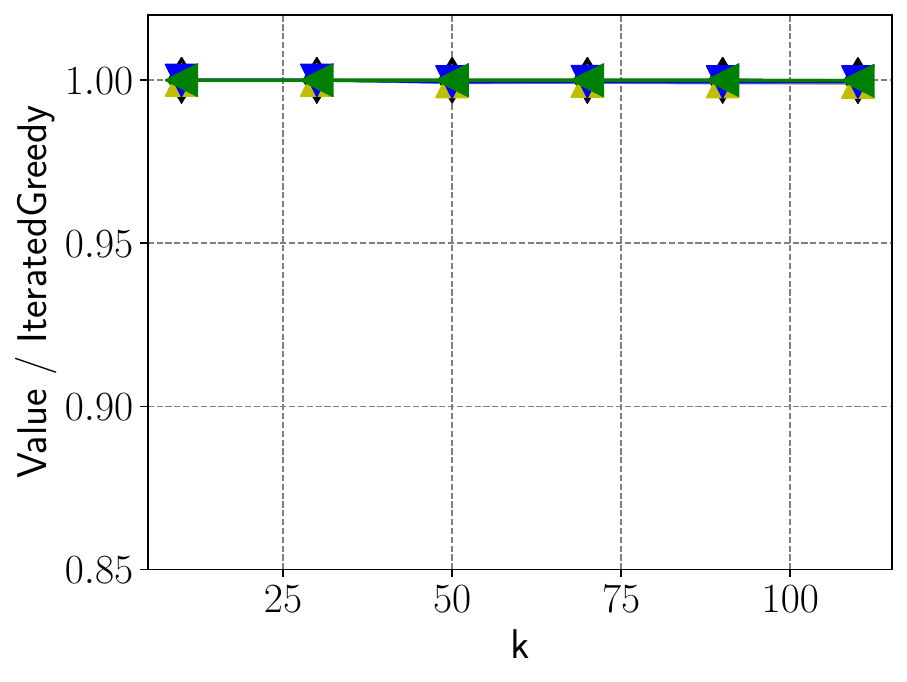}
  }
  \subfigure[Rounds, ca-GrQc]{ \label{fig:rounds-grqc-2-atg}
    \includegraphics[width=0.3\textwidth,height=0.16\textheight]{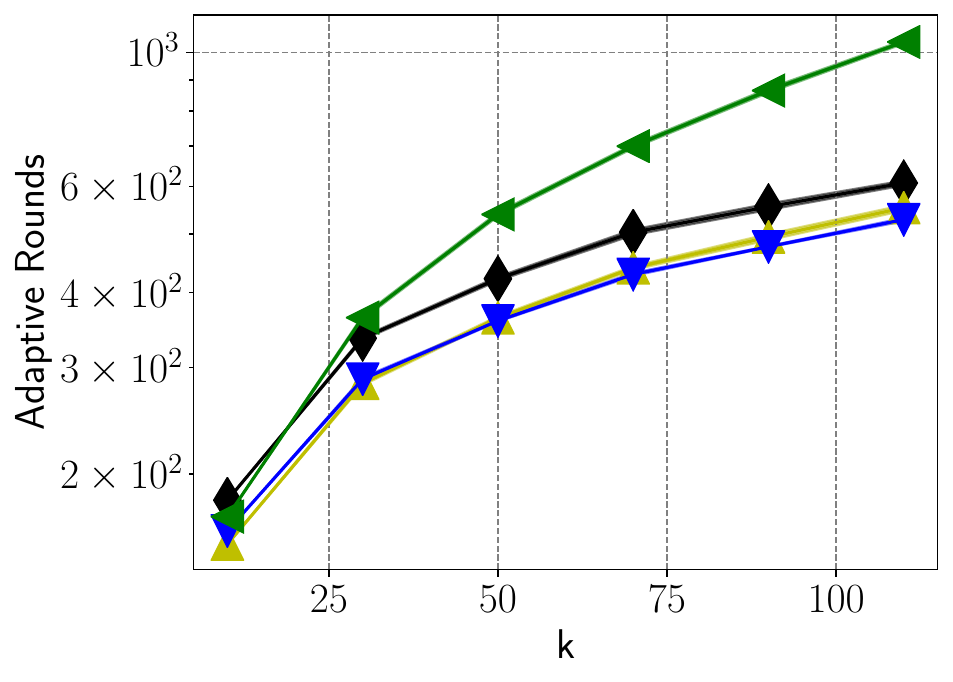}
  }
  \subfigure[Queries, ca-GrQc]{ \label{fig:query-grqc-2-atg}
    \includegraphics[width=0.3\textwidth,height=0.16\textheight]{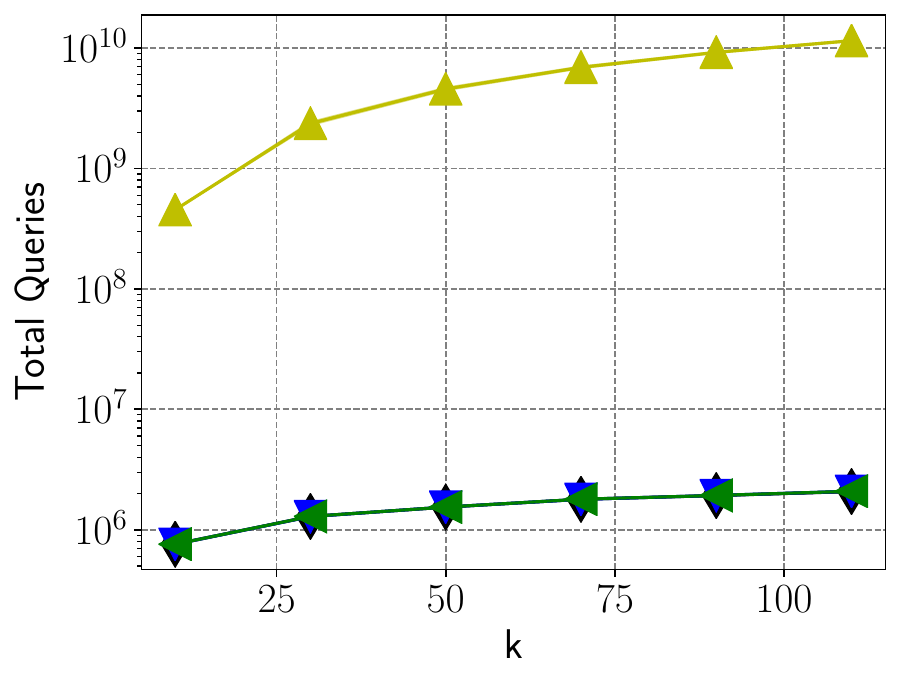}
  }
  \caption{Results of \atg and \latg with four threshold sampling procedures on two datasets. 
  The algorithms are run strictly following pseudocode.
  The legends in Fig.~\ref{fig:val-BA-2-ast} and ~\ref{fig:val-BA-2-atg}
  apply to all other subfigures.
   } \label{fig:main2}
\end{figure*}

\subsection{Comparison of Different Threshold Sampling Procedures.} 
Fig.~\ref{fig:main2} shows the results of \atg and \latg
with different threshold sampling procedures for cardinality-constrained maximum cut 
on two datasets, BA ($n=968$) and ca-GrQc ($n = 5242$).
All the algorithms are run according to pseudocode without any modification.
\threseqama and \tsbin represent the \threseq algorithms without and with
binary search proposed in \citeA{amanatidis2021submodular}.

All four versions of \atg return similar results on objective values;
see Figs.~\ref{fig:val-BA-2-ast} and~\ref{fig:val-grqc-2-ast}.
As for adaptive rounds, \threseq, \thresam, and \threseqama all run in
$\oh{\log(n)}$ rounds, while \tsbin runs in $\oh{\log^2(n)}$ rounds.
By the results in Figs.~\ref{fig:rounds-BA-2-ast} and~\ref{fig:rounds-grqc-2-ast},
our \threseq is the most highly parallelizable algorithm, followed by \threseqama.
\tsbin is significantly worst as what it is in theory.
\red{Perhaps the reason why our algorithm performed better in practical settings 
can be attributed to the following factors.
Theoretically, all algorithms except for \tsbin have the 
same order of adaptivity.
However, the adaptivity of each algorithm is associated with 
different constants, which in turn depend on the design of the algorithm.
Our algorithm maintains two sets $A$, $A' \subseteq A$ during its execution.
At the beginning of each iteration, 
the filtration step is with respect to 
set $A$, which contains elements with negative marginal gains;
at the end of the algorithm, elements with negative marginal gains
are excluded from $A$ to get $A'$,
which is used to get the average marginal gains of the solution.
From an experimental point of view, 
this implementation allows us to filter out more elements after one round
while maintaining the same average marginal gain.}
With respect to the query calls,
while our \threseq only queries once for each prefix,
\thresam queries $16\lceil \log(2/\hat{\delta})/\hat{\epsi}^2 \rceil$ times,
and both \threseqama and \tsbin query $|V|$ times.
According to Figs.~\ref{fig:query-BA-2-ast} and~\ref{fig:query-grqc-2-ast},
our \threseq is the most query efficient one among all.
Also, the total queries do not increase a lot when $k$ increases.
With binary search, \tsbin is the second best one which has $\oh{n\log^2(n)}$
query complexity.
As for \thresam,
with the input values as $n=968$, $k=10$, and $\epsi = \delta = 0.1$,
it queries about $2\times 10^5$ times for each prefix
which is significantly large.

Regarding to different \latg algorithms,
they all return the competitive solutions compared with \iter;
see Figs.~\ref{fig:val-BA-2-atg} and~\ref{fig:val-grqc-2-atg}.
Since each iteration of \latg calls a threshold sampling subroutine
which is based on the solution of previous iterations and 
a slowly decreasing threshold $\tau$,
after the first filtration of the subroutine, the size of the candidate set is limited.
Thus, there is no significant difference between different {\latg}s
concerning rounds and queries.
However, there are two exceptions.
First, since \tsbin is the only one who has $\oh{\log^2(n)}$ adaptive rounds,
it still runs with more rounds; see Figs.~\ref{fig:rounds-BA-2-atg} and~\ref{fig:rounds-grqc-2-atg}.
Also, the number of queries of \latg with \thresam is significantly large
with the same reason discussed before.

Among all, \threseq proposed in this paper is not only the best theoretically,
but also performs well in experiments compared with
the pre-existing threshold sampling algorithms.





\section{Discussion and Future Directions}
\red{In this paper, we propose a new threshold sampling algorithm, \threseq,
which solves \tp on non-monotone instances with high probability, optimal adaptivity,
and query complexity.
Different from other state-of-the-art thresholding algorithms,
\threseq is based on maintaining two sets that separately solve \tp.
Then, we propose two approximation algorithms \algOnefullname
and \algTwofullname that are inspired by \iter.}

\red{Compared to state-of-the-art algorithms,
our \threseq exhibits the highest query efficiency with relatively fewer
adaptive rounds;
\latg produces results that are almost identical to IteratedGreedy in terms of objective value, 
and, relatively speaking, is the most query efficient combinatorial algorithm;
\atg is the second most parallelizable algorithm among all algorithms and
delivers reasonably good objective values.
Despite demonstrating good results, it should be noted that our approximation algorithms rely on
the \unc subroutine, which
requires access to the multilinear extension 
and may be impractical in certain settings.
So, in the experiment, we substituted it with a random subset sampling approach
which provides an expected $(1/4)$-approximation ratio.
However, this substitution may result in a decrease in the objective value during the experiment.
}

\red{Further investigations are needed in our work
and there is still significant room for improvement.
 For instance,
in non-monotone submodular maximization problems,
using the objective value of max singleton to guess $\opt$
is a common practice that involves $\oh{\log(n)}$ guesses.
If we can reduce the number of guesses to a constant, 
the query complexity can be improved significantly by a factor of $\oh{\log(n)}$.
Additionally, the current theoretical best approximation ratio is 0.385. 
\revtwo{In our paper, the best we proposed is a $(0.193-\epsi)-$approximation
algorithm ($(0.139-\epsi)-$approximation in our experiment with a random subset unconstraint algorithm).
Hence, the question remains interesting: Can we parallelize other algorithms that provide a better approximation ratio?
}
}

\revtwo{In our paper, 
we focus on the number of queries and the query parallelizability
assuming that the oracle computation time dominates the overall computation duration.
However, in practical scenarios, the function representation significantly influences algorithm performance (\eg the multilinear extension and its gradient have closed forms for maximum cut).
Thus, for any particular application, there is a lot of room for improvement
based on such specific representation of the submodular function.
}


\clearpage
\appendix
\section{Probability Lemma and Concentration Bounds} \label{apx:prob}
\begin{lemma}\label{lemma:chernoff}
    (Chernoff bounds \cite{mitzenmacher2017probability}). 
    Suppose $X_1$, ... , $X_n$ are independent binary random variables such that 
    $\prob{X_i = 1} = p_i$. Let $\mu = \sum_{i=1}^n p_i$, and 
    $X = \sum_{i=1}^n X_i$. Then for any $\delta \geq 0$, we have
    \begin{align}
        \prob{X \ge (1+\delta)\mu} \le e^{-\frac{\delta^2 \mu}{2+\delta}}.
    \end{align}
    Moreover, for any $0 \leq \delta \leq 1$, we have
    \begin{align}
        \prob{X \le (1-\delta)\mu} \le e^{-\frac{\delta^2 \mu}{2}}.
    \end{align}
\end{lemma}
\begin{lemma} \label{lemma:indep}
    \cite{chen2021best}.
    Suppose there is a sequence of $n$ Bernoulli trials:
    $X_1, X_2, \ldots, X_n,$
    where the success probability of $X_i$
    depends on the results of
    the preceding trials $X_1, \ldots, X_{i-1}$.
    Suppose it holds that $$\prob{X_i = 1 | X_1 = x_1, X_2 = x_2, \ldots, X_{i-1} = x_{i-1} } \ge \eta,$$ where $\eta > 0$ is a constant and $x_1,\ldots,x_{i-1}$ are arbitrary.
  
    Then, if $Y_1,\ldots, Y_n$ are independent Bernoulli trials, each with probability $\eta$ of
    success, then $$\prob {\sum_{i = 1}^n X_i \le b } \le \prob{\sum_{i = 1}^n Y_i \le b }, $$
    where $b$ is an arbitrary integer.
  
    Moreover, let $A$ be the first occurrence of success in sequence $X_i$.
    Then, $$\ex{A} \le 1/\eta.$$
\end{lemma}
\begin{lemma} \label{lemma:indep2}
    \cite{chen2021best}.
    Suppose there is a sequence of $n+1$ Bernoulli trials:
    $X_1, X_2, \ldots,X_{n+1},$
    where the success probability of $X_i$
    depends on the results of
    the preceding trials $X_1, \ldots, X_{i-1}$,
    and it decreases from 1 to 0.
    Let $t$ be a random variable based on the $n+1$ Bernoulli trials.
    Suppose it holds that 
    $$\prob{X_i = 1 | X_1 = x_1, X_2 = x_2, \ldots, X_{i-1} = x_{i-1}, i\le t } \ge \eta,$$ 
    where $x_1,\ldots,x_{i-1}$ are arbitrary and $0 < \eta < 1$ is a constant.
    Then, if $Y_1,\ldots, Y_{n+1}$ are independent Bernoulli trials, each with probability $\eta$ of
    success, then 
    $$\prob {\sum_{i = 1}^t X_i \le bt } \le \prob{\sum_{i = 1}^t Y_i \le bt }, $$
    where $b$ is an arbitrary integer.
\end{lemma}
\section{Counterexample for \thresam with Non-monotone Submodular Functions}
\label{sec:counterexample}
\citeA{Fahrbach2018} proposed a subroutine, 
\thresam, which returns a solution
$S \subseteq \mathcal{N}$ that $\ex{f(S)/S} \ge (1-\epsi)\tau$
within logarithmic rounds and linear time. 
The full pseudocode for \thresam is given in Alg.~\ref{alg:thresh}.
The notation $\mathcal U ( S, t )$ represents the uniform distribution 
over subsets of $S$ of size $t$. 
\thresam relies upon the procedure \reducedmean, given in Alg.~\ref{alg:rm}.
The Bernoulli distribution input to \reducedmean is the distribution 
$\mathcal D_t$, which is defined as follows. 
\begin{definition} \label{def:indicator}
  Conditioned on the current state of the algorithm,
  consider the process where the set $T \sim \mathcal U (A, t - 1)$ and then
  the element $x \sim A \setminus T$ are drawn uniformly at random.
  Let $\mathcal D_t$ denote the probability distribution over the indicator
  random variable 
  $$I_t = \mathbb I [f( S \cup T  + x) - f( S \cup T ) \ge \tau ]. $$
\end{definition}

Below, we state the lemma of \thresam in \citeA{Fahrbach2018}.
\begin{lemma}[\citeS{Fahrbach2018}] \label{lemm:thresh} 
  The algorithm \thresam outputs $S \subseteq \mathcal N$ with
  $|S| \le k$ in $\oh{\log (n / \delta ) / \epsi }$ adaptive
  rounds such that the following properties hold with 
  probability at least $1 - \delta$: 
  \begin{itemize}
    \item[1.] There are $\oh{n / \epsi}$ oracle queries in expectation.
    \item[2.] The expected average marginal $\ex{f(S)/|S|} \ge (1-\epsi)\tau$.
    \item[3.]  If $|S| < k$, then $f_x(S) < \tau$ for all 
    $x \in \mathcal N$.
  \end{itemize}
\end{lemma}

In \citeA{Fahrbach2018a} and \citeA{kuhnle2021nearly},
the above lemma is used with non-monotone submodular functions;
however, in the case that $f$ is non-monotone, the lemma does not hold.
Alg.~\ref{alg:rm} only checks (on Line \ref{line:Check}) if there is 
more than a constant fraction of 
elements whose marginal gains are larger than the threshold $\tau$.
If there exist elements with
large magnitude, \textit{negative} marginal gains, then
the average marginal gain may fail to satisfy
the lower bound in Lemma~\ref{lemm:thresh}.
As for the proof in \citeA{Fahrbach2018a},
the following inequality does not hold (needed for the proof of Lemma
3.3 of \citeA{Fahrbach2018a}):
\[\ex{\marge{T}{S}}\ge (\ex{I_1}+\ex{I_2} +\ldots +\ex{I_t})\tau,\]
where $|T| = t^*$ and $t\ge t^*/(1+\hat{\epsi})$.
Next, we give a counterexample for the two versions of \thresam used in 
\citeA{Fahrbach2018} and \citeA{Fahrbach2018a}
where the only difference is that
the if condition in Alg.~\ref{alg:thresh}
on Line~\ref{line:stopWithA} changes to $|A|<3k$
in \citeA{Fahrbach2018a}.
\begin{algorithm}[t] 
  \caption{The \reducedmean algorithm of \citeA{Fahrbach2018}} \label{alg:rm}
  \begin{algorithmic}[1]
  \State \textbf{Input:} access to a Bernoulli distribution $\mathcal D$,
 error $\epsi$, failure probability $\delta$
 \State Set number of samples $m \gets 16 \lceil \log (2 / \delta ) / \epsi^2 \rceil$
 \State Sample $X_1, X_2, \ldots, X_m \sim \mathcal D$
 \State Set $\bar \mu \gets \frac{1}{m} \sum_{i = 1}^m X_i$
 \If{ $\bar \mu \le 1 - 1.5 \epsi$ }
 \State \textbf{return} \texttt{true}
 \EndIf
 \State \textbf{return} \texttt{false}
 \end{algorithmic}
 \end{algorithm}
 \begin{algorithm}[t]
  \caption{The threshold sampling algorithm of \citeA{Fahrbach2018}}
  \label{alg:thresh}
  \begin{algorithmic}[1]
    \Procedure{\thresam}{$f, k, \tau, \epsi,\delta$}
    \State \textbf{Input:} evaluation oracle $f:2^{\mathcal N} \to \reals$, constraint $k$,
    threshold $\tau$, error $\epsi$, failure probability $\delta$
    \State Set smaller error $\hat{\epsi} \gets \epsi / 3$
    \State Set iteration bounds $r \gets \lceil \log_{(1 - \hat \epsi)^{-1}}(2n / \delta) \rceil, m \gets \lceil \log(k) / \hat \epsi \rceil$
    \State Set smaller failure probability $\hat \delta \gets \delta / (2r(m + 1))$
    \State Initialize $S \gets \emptyset, A \gets N$
    \For{ $r$ sequential rounds }
    \State Filter $A \gets \{ x \in A : \Delta (x, S) \ge \tau \}$\label{line:filter}
    \If{$|A| = 0$}\label{line:stopWithA}
    \State \textbf{break}
    \EndIf
    \For{$i = 0$ to $m$ in parallel}\label{thresh:for}
    \State Set $t \gets \min \{ \lfloor (1 + \hat \epsi)^i \rfloor, |A|\}$
    \State $rm[t] \gets $\reducedmean $( \mathcal D_t, \hat \epsi, \hat \delta )$ \label{line:Check}
    \EndFor
    \State $t' \gets \min t$ such that $rm[t]$ is \texttt{true}
    \State Sample $T \sim \mathcal U \left( A, \min \{t', k - |S| \} \right)$
    \State Update $S \gets S \cup T$
    \If{ $|S| = k$ }
    \State \textbf{break}
    \EndIf
    \EndFor
    \State \textbf{return} $S$
    \EndProcedure
\end{algorithmic}
\end{algorithm}

\textbf{Counterexample 1.}
Define a set function $f: 2^{\uni} \to \reals$ as follows,
\red{where $a \in \uni$,}
\begin{equation*}
  f(B)= \left\{
    \begin{aligned}
      &n^2 + |B|,&\text{if } a \not \in B\\
      &n^2 +1- (|B|-1)n, &\text{if }a \in B
    \end{aligned}
  \right. .
\end{equation*}
Let $k = n = |\mathcal{N}|>400$, $\tau=1$, $\epsi=0.1$, $\delta=0.1$.
Run $\thresam(f, k, \tau, \epsi, \delta)$.
\begin{proof}
For any $B \subseteq \mathcal{N}$ and $x\in \mathcal{N} \backslash B$, 
the above set function follows that
\begin{equation*}
  \marge{x}{B}= \left\{
    \begin{aligned}
      &1,&&\text{if } x \neq a \text{ and } a \not \in B\\
      &-n,&&\text{if } x \neq a \text{ and } a \in B\\
      &1-|B|(n+1), &&\text{if } x = a
    \end{aligned}
  \right. .
\end{equation*}
Thus, $f$ is a non-negative, non-monotone submodular function.

\red{Consider the first iteration of the outer for loop, where $S = \emptyset$,
and $A = \uni$ after Line~\ref{line:filter}.}
For any $1 < t \le |\mathcal{N}|$, $|T| = t-1$,
\[\ex{I_t} = \prob{f( S \cup T  + x) - f( S \cup T ) \ge \tau }
= \prob{x \neq a \text{ and } a \not \in T}
= 1-\frac{t}{n}.\]
So, with any value of $\epsi$, 
\reducedmean returns \texttt{true} when $t> \epsi n/2$.
The first round of \thresam
samples a set $T_1$ with $t_1'=|T_1| > \epsi n/2$.
Then update $S$ by $S = T_1$.

For the \thresam in \citeA{Fahrbach2018a} with stop condition
$|A| < 3k$, the algorithm stopped here after the first iteration,
no matter what is sampled.
In this case, the expectation of marginal gains of the set returned
by the algorithm would be as follows,
\red{
\begin{align*}
  \ex{\marge{S}{\emptyset}} &= \prob{a \in T_1}\cdot \exc{\marge{T_1}{\emptyset}}{a \in T_1}
  +\prob{a \not \in T_1}\cdot\exc{\marge{T_1}{\emptyset}}{a \not \in T_1}\\
  &= \frac{t_1'}{n}\cdot\left(1-(t_1'-1)n\right) + \frac{n-t_1'}{n}\cdot t_1' \\
  &= t_1'\left(2-t_1'+\frac{1-t_1'}{n}\right) < 0.
\end{align*}}
Next, we consider the \thresam with stop condition $|A|=0$.
After the first iteration discussed above,
if $a \in T_1$, all the elements would be filtered out at the second round.
Algorithm stoped here and returned $S$, say $S_1$.
If $a \not \in T_1$, $T_1$ and $a$ would be filtered out at the second round,
which means $A = \mathcal{N} \backslash (S \cup \{a\})$.
And for any $T \subseteq A $ and 
$x \in A \backslash T$, 
\[f(S \cup T + x) - f(S \cup T) = 1.\]
Therefore, $\ex{I_t} = 1$ for all $t$.
After several iterations, $S = \mathcal{N} \backslash \{a\}$ would be returned, say $S_2$.

The expectation of objective value of the set returned would be as follows,
\red{\begin{align*}
  \ex{\marge{S}{\emptyset}}&= \prob{a \in T_1}\cdot \exc{\marge{S_1}{\emptyset}}{a \in T_1}
  + \prob{a \not \in T_1} \cdot \exc{\marge{S_2}{\emptyset}}{a \not \in T_1}\\
  & = \frac{t_1'}{n}(1 - (t_1'-1)n) + \frac{n-t_1'}{n}(n-1)\\
  &= \frac{2t_1'}{n}-1-t_1'^2+n < 0,
\end{align*}}
since $\epsi = 0.1$, $n > 400$, and $\epsi n/2 < t_1' < \epsi(1+\epsi/3)n/2$.
\end{proof}

\section{Proofs for Section~\ref{sec:ts}} \label{apx:threseq}
\ThresholdFilterSet*
\begin{proof}[proof of Lemma~\ref{lemma:ThresholdFilterSet}]
    After filtering on Line~\ref{line:threshold-filtering},
    any element $x \in V$ follows that $\marge{x}{A} \ge \tau$.
    Therefore, $S_0 = \emptyset$.
    Also, it is obvious that $\marge{x}{A\cup V}=0$.
    So, $S_{|V|} = V$.
    Next, let's consider any $x \in S_{i-1}$.
    By submodularity,
    $$\marge{x}{A \cup T_i} \le \marge{x}{A\cup T_{i-1}}< \tau.$$
    Thus, for any $x \in S_{i-1}$, it holds that $x \in S_i$,
    which means $S_{i-1} \subseteq S_i$.
\end{proof}
\ThresholdProb*
\begin{proof}[proof of Lemma~\ref{lemma:ThresholdProb}]
    Call an element $v_i \in V$ \textit{bad} iff
    $\marge{v_i}{A\cup T_{i-1}} < \tau$;
	and \textit{good}, otherwise.
    The random permutation of $V$ can be regarded
    as $|V|$ dependent Bernoulli trials, with success
    iff the element is bad and failure otherwise.
	Observe that, the probability that an element in $T_i$ is bad,
	when $i \le t$,
	is less than $\epsi/2$, conditioned on the outcomes of the 
	preceding trials.
    We know that,
    $$\prob{i^* < \min\{s,t\}} \le \prob{\#\text{ bad elements in }
    T_{i'} > \epsi i' \text{, where } i' = \min\{s,t\}}.$$
    Let $X_i=1$, if $v_i$ is bad; and $X_i=0$, otherwise.
    Then, $(X_i)$ is a sequence of
    dependent Bernoulli trails.
    And for any $i \le i'$,
    $\prob{X_i = 1} \le \epsi/2$.
    Let $(Y_i)$ be a sequence of independent and identically distributed Bernoulli trails,
    each with success probability $\epsi/2$.
    Then, the probability of $i^* < \min\{s,t\}$ can be bounded as follows:
    $$\prob{i^* < \min\{s,t\}} 
    \le \prob{\sum_{i=1}^{i'} X_i > \epsi i'} \overset{(a)}{\le}
    \prob{\sum_{i=1}^{i'} Y_i > \epsi i'} \overset{(b)}{\le} 1/2, $$
    where Inequality (a) follows from Lemma~\ref{lemma:indep2},
    and Inequality (b) follows from Law of Total Probability
    and Markov's inequality.
\end{proof}

\ThresholdGood*
\begin{proof}[proof of Lemma~\ref{lemma:ThresholdGood}]
    Let $A_j$ be the set $A$ after iteration $j$,
    $T_{j,i}$ be the first $i$ elements of $V_j$ at $j$-th iteration.
    Similarly, define $A_j'$ as the set $A'$ after iteration $j$,
    \red{$T'_{j,i} = T_{j,i} \cap A'$}.

    From Algorithm~\ref{alg:threshold}, $A=\sum_{j=1}^{\ell} T_{j,i^*}$.
    For each $T_{j,i^*}$, there are at least $(1-\epsi)$-fraction of $T_{j,i^*}$ are good.
    Totally, there are at least $(1-\epsi)$-fraction of $A$ are good.

    By Line~\ref{line:threshold-Aprime}, $T'_{j,i}$ only contains the elements
    with nonnegative marginal gains in $T_{j,i}$.
    Therefore, any element in $A'$ has nonnegative marginal gain when added.
    For any good element $v_i \in V_j$, by submodularity, 
    $\marge{v_i}{A'_{j-1}\cup T'_{j,i-1}}\ge 
    \marge{v_i}{A_{j-1}\cup T_{j,i-1}} \ge \tau$.
    Thus, a good element in $A$ is always good in $A'$.
\end{proof}

\section{\threshgreedy and Modification} \label{apx:threshgreedy}
\begin{algorithm}[t]
   \caption{The \threshgreedy Algorithm of \citeA{Badanidiyuru2014}}
   \label{alg:threshgreedy}
   \begin{algorithmic}[1]
     \Procedure{\threshgreedy}{$f, k, \epsi$}
     \State \textbf{Input:} evaluation oracle $f:2^{\mathcal N} \to \reals$, constraint $k$,
     accuracy parameter $\epsi > 0$
     \State $M \gets \argmax_{x \in \mathcal N} f(x)$; 
     \State $S \gets \emptyset$
     \For{ $\tau = M$; $\tau \ge (1 - \epsi)M/k$; $\tau \gets \tau (1 - \epsi)$}
     \For{ $x \in \mathcal N$ }
     \If{ $f( S \cup \{x\} ) - f(S) \ge \tau$ }
     \State $S \gets S \cup \{x\}$
     \If{ $|S| = k$ }
     \State \textbf{break} from outer \textbf{for}
     \EndIf
     \EndIf
     \EndFor
     \EndFor
     \State \textbf{return} $S$
     \EndProcedure
\end{algorithmic}
\end{algorithm}
In this section, we describe \threshgreedy (Alg.~\ref{alg:threshgreedy}) 
of \citeA{Badanidiyuru2014} and
how it is modified to have low adaptivity. This algorithm achieves 
ratio $1 - 1/e - \epsi$ in $\oh{n \log k}$ queries
if the function $f$ is monotone but has no constant ratio
if $f$ is not monotone. 

The \threshgreedy algorithm works as follows:
a set $S$ is initialized to the empty set. Elements
whose marginal gain exceed a threshold value are added
to the set in the following way:
initially,
a threshold of $\tau = M = \argmax_{a \in \mathcal N} f(a)$
is chosen, which is iteratively decreased by a factor
of $(1 - \epsi)$ until $\tau < M/k$. For each threshold $\tau$,
a pass through all elements
of $\mathcal N$ is made, during which any
element $x$ that satisfies $f(S \cup \{x\}) - f(S) \ge \tau$  is added to the set $S$.
While this strategy leads to an efficient $\oh{n \log k}$ total
number of queries, it also has $\Omega (n \log k)$ adaptivity, as
each query depends on the previous ones.

To make this approach less adaptive,
we replace the highly adaptive pass through $\mathcal N$ (the inner \textbf{for} loop)
with a single call to \thresam, which requires $\oh{\log n}$
adaptive rounds and $\oh{n / \epsi}$ queries in expectation.
This modified greedy approach appears twice in \latg (Alg.~\ref{alg:latg}),
corresponding to the two \textbf{for} loops. 
\section{Improved Ratio for \iter} \label{apx:iter}
\begin{algorithm}[t]
   \caption{The \iter Algorithm of \citeA{Gupta2010a}}
   \label{alg:iter}
   \begin{algorithmic}[1]
     \Procedure{\iter}{$f, k$}
     \State \textbf{Input:} evaluation oracle $f:2^{\mathcal N} \to \reals$, constraint $k$
     \State $A \gets \emptyset$
     \For{ $i \gets 1$ to $k$ }
     \State $a_i \gets \argmax_{x \in \mathcal N} f( A \cup \{x\} ) - f(A)$
     \State $A \gets A \cup \{a_i\}$
     \EndFor
     \State $B \gets \emptyset$
     \For{ $i \gets 1$ to $k$ }
     \State $b_i \gets \argmax_{x \in \mathcal N \setminus A} f ( B \cup \{x\}) - f (B)$
     \State $B \gets B \cup \{b_i\}$
     \EndFor
     \State $A' \gets $ \unc $(A)$
     \State \textbf{return} $C \gets \argmax \{ f(A), f(A'), f(B) \}$
     \EndProcedure
\end{algorithmic}
\end{algorithm}
In this section, we prove an improved approximation ratio
for the algorithm \iter of \citeA{Gupta2010a}, wherein a
ratio of $1/(4 + \alpha)$ is proven given access to
a $1/\alpha$-approximation for \unc. 
We improve this ratio to $\iterratio \approx 0.193$ if $\alpha = 2$.
Pseudocode for \iter
is given in Alg.~\ref{alg:iter}.

\iter works as follows. First a standard greedy procedure is run
which produces set $A$ of size $k$.
Next, a second greedy procedure is run to yield set $B$; during
this second procedure, elements of $A$ are ignored.
A subroutine for \unc is used on $f$ restricted to $A$, which
yields set $A'$. Finally the set of $\{A, A', B \}$ that maximizes
$f$ is returned. 
\begin{theorem}
  Suppose there exists an $(1/\alpha)$-approximation for
  \unc. Then by using this procedure as a subroutine,
  the algorithm \iter has approximation ratio
  $\iterratio$ for \sm.
\end{theorem}
\begin{proof}
  For $1 \le i \le k$, let $a_i,b_i$ be as chosen during the
  run of \iter. Define $A_i = \{ a_1, \ldots, a_{i - 1} \}$,
  $B_i = \{ b_1, \ldots, b_{i - 1} \}$. Then
  for any $1 \le i \le k$, we have
  \begin{align*}
    f(A_{i + 1}) + f(B_{i + 1}) - f(A_i) - f(B_i) &= f_{a_i}(A_i) + f_{b_i}( B_i ) \\
 &\ge \frac{1}{k} \sum_{o \in O} f_o( A_i ) + \frac{1}{k} \sum_{o \in O \setminus A} f_o (B_i ) \\
    &\ge \frac{1}{k} \left( \ff{O \cup A_i} - \ff{A_i} + \ff{ (O \setminus A) \cup B_i } - \ff{ B_i }\right) \\
    &\ge \frac{1}{k} \left( \ff{ O \setminus A } - (\ff{A_i} + \ff{B_i}) \right),
  \end{align*}
  where the first inequality follows from the greedy choices,
  the second follows from submodularity,
  and the third follows from submodularity and the fact that $A_i \cap B_i = \emptyset$.
  Hence, from this recurrence and standard arguments, 
  $$f(A) + f(B) \ge (1 - 1/e) \ff{O \setminus A},$$
  where $A,B$ have their values at termination of \iter.
  Since $f(A') \ge f(O \cap A) / \alpha$, 
  we have from submodularity
  \begin{align*}
    f(O) &\le f( O \cap A ) + f(O \setminus A) \\
    &\le \alpha f(A') + (1-1/e)^{-1}(f(A) + f(B)) \\
    &\le ( \alpha + 2(1-1/e)^{-1} ) f(C).  &\qedhere
  \end{align*}
\end{proof}
\section{Proofs for Section~\ref{sec:latg}} \label{apx:latg}
In this section, we provide the
proofs omitted from Section~\ref{sec:latg}.

\LemmaA*
\begin{proof}[Proof of Lemma~\ref{lemm:A}]
  Since each element in $A'$ has nonnegative marginal gain, 
  it always holds that 
  $f(\mathcal{A}_j')\ge f(\mathcal{A}_{j-1}')$.

  From Lemma~\ref{lemma:ThresholdGood}, there are at least $(1-\epsi)$-fraction
  of $A$ are good elements.
  Therefore, there are at least $(1-\epsi)k$ of $a_j'$
  which is good element or dummy element.
  Next, let's consider the following 3 cases of $a_j'$.

  \textbf{Case} $i(j)=1$ and $a_j'$ is good.
  By Theorem~\ref{thm:threshold} and Lemma~\ref{lemma:ThresholdGood}, it holds that 
  $$f(\mathcal{A}_j')-f(\mathcal{A}_{j-1}')\ge \tau_1=M \ge \frac{1}{k} \sum_{o \in O}f(o)\ge \frac{1}{k}f(O).$$

  \textbf{Case} $i(j)>1$ and $a_j'$ is good.
  Since $a_j'$ is returned at iteration $i(j)$ and $a_j'$ is good, it holds that:
  (1) $f(\mathcal{A}_j')-f(\mathcal{A}_{j-1}')\ge \tau_{i(j)}$; 
  (2) at previous iteration $i(j)-1$, \threseq returns $S_{i(j)-1}$ that $|S_{i(j)-1}| < k-|A_{i(j)-2}|$.
  By property (2) and Theorem~\ref{thm:threshold}, for any $o \in O\backslash A_{i(j)-1}$,
  $\marge{o}{A_{i(j)-1}} < \tau_{i(j)-1}$.
  Then,
  \begin{align*}
    f(\mathcal{A}_j')-f(\mathcal{A}_{j-1}')&\ge \tau_{i(j)}= (1-\epsi')\tau_{i(j)-1}\\
    &> \frac{1-\epsi'}{k}\sum_{o\in O\backslash A_{i(j)-1}}\marge{o}{A_{i(j)-1}}\\
    &\ge \frac{1-\epsi'}{k}\left(f(O\cup A_{i(j)-1})-f(A_{i(j)-1})\right)\\
    &\ge \frac{1-\epsi'}{k}\left(f(O\cup A_{i(j)-1})-f(A'_{i(j)-1})\right) \numberthis \label{ineq:sub8}\\
    &\ge \frac{1-\epsi'}{k}\left(f(O\cup A_{i(j)-1})-f(\mathcal{A}_{j-1}')\right), \numberthis \label{ineq:sub7}
  \end{align*}
  where Inequality~\ref{ineq:sub8} follows from the proof of Lemma~\ref{lemma:ThresholdGood},
  and Inequality~\ref{ineq:sub7} follows from $A'_{i(j)-1} \subseteq \mathcal{A}_{j-1}'$.

  \textbf{Case} $i(j)=\ell+1$ (or $a_j'$ is dummy element). In this case, 
  $|A|<k$ when the first \textbf{for} loop ends.
  So, \threseq in the last iteration returns $S_\ell$ that $|S_\ell| < k-|A_{\ell-1}|$.
  From Theorem~\ref{thm:threshold}, it holds that $\marge{o}{A_{\ell}} < \tau_\ell < \frac{M}{ck}$, 
  for any $o \in O\backslash A_{\ell}$.
  Thus,
  $$\frac{M}{ck}> \frac{1}{k}\sum_{o\in O\backslash A_{\ell}} \marge{o}{A_{\ell}}\ge 
  \frac{1}{k}\left(f(O\cup A_\ell) - f(A_\ell)\right) \overset{(a)}{\ge} 
  \frac{1}{k}\left(f(O\cup A_\ell) - f(\mathcal{A}_j')\right),$$
  where Inequality (a) follows from $A_\ell = A$ and $f(\mathcal{A}_j') = f(A')$.

  The first inequality of Lemma~\ref{lemm:A} holds in those three cases
  with at least $(1-\epsi')k$ of $j$.
\end{proof}

\LemmaC*
\begin{proof}[Proof of Lemma~\ref{lemma:three}]
  From Lemma~\ref{lemm:A}, $f(\mathcal{A'}_{j(u)-1})\ge f(\mathcal{A'}_{j(u-1)})$, and
  \begin{align*}
    f(\mathcal{A}_{j(u)}')& \ge \left(1-\frac{1-\epsi'}{k}\right)f(\mathcal{A}_{j(u)-1}')+
    \frac{1-\epsi'}{k}f(O\cup A_{i(j(u))-1})-\frac{M}{ck}\\
    &\ge \left(1-\frac{1-\epsi'}{k}\right)f(\mathcal{A}_{j(u-1)}')+
    \frac{1-\epsi'}{k}f(O\cup A_{i(j(u))-1})-\frac{M}{ck}.
  \end{align*}
  Similarly,
  $$f(\mathcal{B}_{j(u)}')\ge \left(1-\frac{1-\epsi'}{k}\right)f(\mathcal{B}_{j(u-1)}')+
  \frac{1-\epsi'}{k}f((O\backslash A)\cup B_{i(j(u))-1})-\frac{M}{ck}.$$
  By adding the above two inequalities and the submodularity, we have,
  \begin{align*}
    \Gamma_u&\ge \left(1-\frac{1-\epsi'}{k}\right) \Gamma_{u-1}+
    \frac{1-\epsi'}{k}\left(f(O\cup A_{i(j(u))-1})+f((O\backslash A)\cup B_{i(j(u))-1})\right)-\frac{2M}{ck}\\
    &\ge \left(1-\frac{1-\epsi'}{k}\right) \Gamma_{u-1}+\frac{1-\epsi'}{k} f(O\backslash A)-\frac{2M}{ck}.&\qedhere
  \end{align*}
\end{proof}

\begin{lemma} \label{lemm:seven}
  Let $\epsi \in (0,1)$, and suppose
  $c = 8 / \epsi$, $\epsi' = (1 - 1/e) \epsi / 8$,
  and $\beta = 1 - e^{(1 - \epsi')^2}$.
  Then
  \begin{equation}\left( \frac{1 - \frac{2}{c(1-\epsi')}}{\alpha + 
    \frac{2}{\beta}} \right) \ge \left( \frac{e - 1}{\alpha(e - 1) + 2e} - \epsi  
    \right). \label{ineq:ana}
  \end{equation}
\end{lemma}
\begin{proof}[Proof of Lemma~\ref{lemm:seven}]
  We start with the following two inequalities, which are verified below.
  \begin{align}
    1 - \frac{2}{c(1-\epsi')} \ge 1 - \epsi& ,\label{ineq:5} \\
    \frac{2}{1 - e^{-(1 - \epsi')^2}} \le \frac{2}{1 - 1/e} + \epsi / 2& .\label{ineq:7}
  \end{align}
  Let $A = 1$, $B = 1 / \alpha + 2/(1 - 1/e)$. 
  From the inequalities above,
  the left-hand side of (\ref{ineq:ana})
  is at least $\frac{A - \epsi}{B + \epsi}$ and
  \begin{align*}
    \frac{A - \epsi}{B + \epsi} \ge \frac{A}{B} - \epsi &\iff \epsi \ge \frac{A}{B} - \frac{A - \epsi}{B + \epsi} \\
                                                     &\iff 1 \ge \frac{A}{\epsi B} - \frac{A}{\epsi (B + \epsi )} + \frac{1}{B + \epsi}.
  \end{align*}
  Next,
  \begin{align*}
    \frac{A}{\epsi B} - \frac{A}{\epsi (B + \epsi )} + \frac{1}{B + \epsi} = \frac{A}{B (B + \epsi )} + \frac{1}{B + \epsi} \le 1/4 + 1/2 < 1,
  \end{align*}
  since $B \ge 2$ and $A = 1$. Finally, $A/B = \frac{\alpha (e - 1)}{e - 1 + 2\alpha e}$.

  \textbf{Proof of (\ref{ineq:5}).}
  \begin{align*}
    c \ge 8 / \epsi  \ge \frac{2}{\epsi(1 - \epsi')},
  \end{align*}
  since $\epsi' = (1 - 1/e) \epsi / 8 < 3/4$.


  \textbf{Proof of (\ref{ineq:7}).}
  Let $\lambda = 1 - 1/e$, $\kappa = e^{-(1-\epsi' )^2}$.
  Inequality (\ref{ineq:7})
  is satisfied iff.
  \begin{align*}
    2 \lambda \le 2(1 - \kappa ) + \frac{\lambda \epsi (1 - \kappa )}{2} &\iff 2\lambda \le 2 - 2 \kappa + \lambda \epsi / 2 - \lambda \epsi \kappa / 2 \\
                                                                        &\iff 2 \kappa + \lambda \epsi \kappa /2 \le \lambda \epsi /2 + 2 - 2 \lambda \\
                                                                        &\iff \kappa = e^{-(1 - \epsi')^2} \le \frac{\lambda \epsi / 2 + 2 - 2 \lambda }{2 + \lambda \epsi / 2} \\
    &\iff (1 - \epsi' )^2 \ge \log \left( \frac{2 + \lambda \epsi / 2}{\lambda \epsi / 2 + 2 - 2 \lambda} \right),
  \end{align*}
  which in turn is satisfied if
  $$2 \epsi' \le 1 - \log \left( \frac{2 + \lambda \epsi / 2}{\lambda \epsi / 2 + 2 - 2 \lambda} \right).$$
  Then
  \begin{align*}
    2 \epsi' = \lambda \epsi / 4 \le \frac{2 + \lambda \epsi /2 - 4 \lambda}{2 - \lambda} &\le \frac{2 + \lambda \epsi / 2 - 4 \lambda}{2 + \lambda \epsi / 2 - 2 \lambda} \\
                                                                                     &= \frac{2 ( \lambda \epsi /2 + 2 - 2 \lambda ) - 2 - \lambda \epsi / 2}{\lambda \epsi /2 + 2 - 2 \lambda} \\
                                                                                     &= 2 - \frac{2 + \lambda \epsi / 2}{\lambda \epsi /2 + 2 - 2\lambda} \\
                                                                                     &= 1 - \left( \frac{2 + \lambda \epsi / 2}{\lambda \epsi /2 + 2 - 2\lambda} - 1 \right) \\
                                                                                     &\le 1 - \log \left( \frac{2 + \lambda \epsi / 2}{\lambda \epsi /2 + 2 - 2\lambda} \right),
  \end{align*}
  where we have used $\log x \le x - 1$, for $x > 0$.
\end{proof}

\section{Multilinear Extension and Implementation of \citeA{Ene2020}} \label{apx:ene}
In this section, we describe the multilinear extension and implementation of 
\citeA{Ene2020}. The multilinear extension $F$
of set function $f$ is defined to be, for $x \in [0,1]^n$:
$$ F(x) = \ex{ f(S) } = \sum_{S \subseteq V} f(S) Pr(S), $$
where $$Pr(S) = \prod_{i \in S} x_i \cdot \prod_{i \not \in S} (1 - x_i).$$
The gradient is approximated by using the central difference in each coordinate
$$ \frac{dF}{dx}(x) \approx \frac{F( x + \gamma / 2 ) - F(x - \gamma /2)}{\gamma}, $$
unless using this approximation required evaluations outside the unit cube, in which
case the forward or backward difference approximations were used. 
The parameter $\gamma$ is set to $0.5$.

Finally, for the maximum cut application, closed forms expressions exist for
both the multilinear extension and its gradient. These are:
$$ F(x) = \sum_{(u,v) \in E} x_u \cdot (1 - x_v ) + x_v \cdot (1 - x_u), $$
and
$$(\nabla F)_u  = \sum_{v \in N(u)} (1 - 2x_v). $$

\textbf{Implementation.} The algorithm was implemented as specified in the pseudocode on page
19 of the arXiv version
of \citeA{Ene2020}. We followed the same parameter choices as in
\citeA{Ene2020}, although we set $\epsi = 0.1$ as setting it to $0.05$
did not improve the objective value significantly but caused a large increase
in runtime and adaptive rounds. The value of $\delta = \epsi^3$ was used
after communications with the authors. 
\begin{figure*}[t] \centering
  \subfigure[Objective value]{ \label{fig:val-Ene}
    \includegraphics[width=0.28\textwidth,height=0.15\textheight]{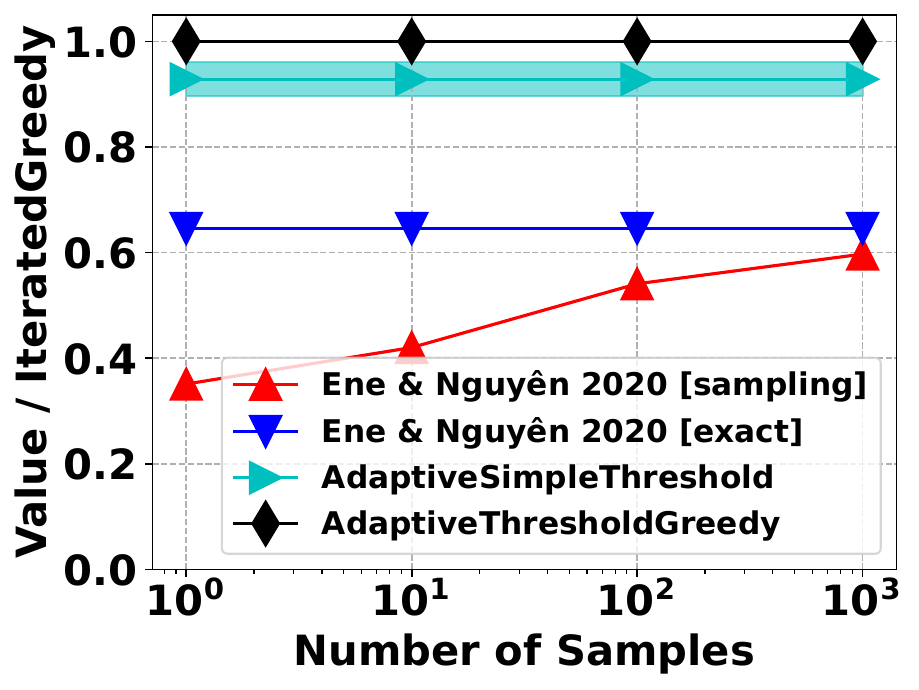}
  }
  \subfigure[Queries to set function]{ \label{fig:queryEne}
    \includegraphics[width=0.28\textwidth,height=0.15\textheight]{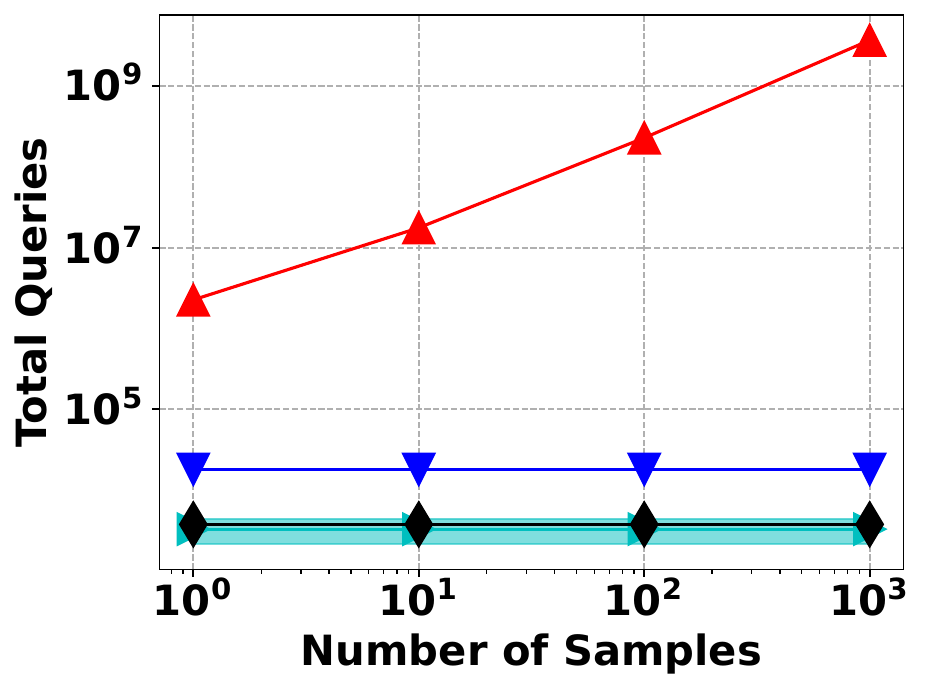}
  }
  \subfigure[Adaptive Rounds]{ \label{fig:roundsEne}
    \includegraphics[width=0.28\textwidth,height=0.15\textheight]{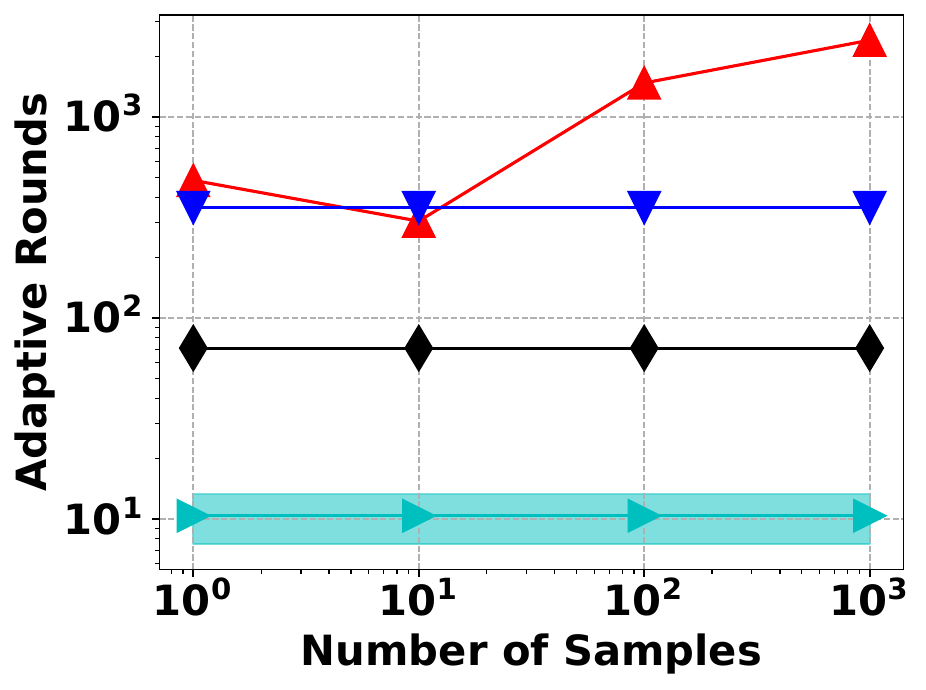}
  }
  \caption{Comparison of our algorithms with \citeA{Ene2020} on a very small random graph ($n = 87$, $k = 10$). In all plots, the $x$-axis shows the number of samples used to approximate the multilinear extension. } \label{fig:multilinear}
\end{figure*} 
\subsection{Additional Experiments}\label{apx:exp}
In this section, we further investigate the performance of
\citeA{Ene2020} when closed-form evaluation of the multilinear
extension and its gradient are impossible. \red{It is known}
that sampling to approximate the multilinear extension
and its gradient is extremely inefficient or yields poor solution
quality with a small number of samples. For this reason, we exclude
this algorithm from our revenue maximization experiments.
To perform this evaluation, we compared versions of the
algorithm of \citeA{Ene2020} that use varying number
of samples to approximate the multilinear extension.

Results are shown in Fig.~\ref{fig:multilinear} 
on a very small random graph with $n=87$ and $k = 10$.
The figure shows the objective value
and total queries to the set function vs. the number of samples
used to approximate the multilinear extension. 
There is a clear
tradeoff between the solution quality and the number of queries
required; at $10^3$ samples per evaluation, the algorithm
matches the objective value of the version with the 
exact oracle;
however, even at roughly $10^{11}$ queries (corresponding
to $10^4$ samples for each evaluation of the multilinear extension),
the algorithm of \citeA{Ene2020} is unable to exceed $0.8$ of
the IteratedGreedy value. On the other hand, if $\le 10$ samples are used to
approximate the multilinear extension, the algorithm is unable to exceed
$0.5$ of the IteratedGreedy value and still requires on the order of $10^7$
queries. 
\bibliography{main.bib}
\bibliographystyle{theapa}

\end{document}